\documentclass[conference]{IEEEtran}

\IEEEoverridecommandlockouts

\usepackage{graphicx}
\usepackage{subcaption}
\usepackage{tikz,xparse}

\newcommand\lineWidth{0.75pt}
\newcommand\markSize{2.5}
\usetikzlibrary{patterns,dsp,chains,matrix,calc,backgrounds,circuits.logic.US,mux,fit}
\pgfdeclarelayer{background}
\pgfdeclarelayer{foreground}
\pgfsetlayers{background,main,foreground}

\usepackage{mathptmx}
\usepackage{verbatim}
\usepackage{calc}%
\usepackage{ifthen}
\usepackage{xifthen}
\usepackage{cancel}
\usepackage{bm}
\usepackage{verbatim}
\usepackage{multirow}
\usepackage{cite}
\usepackage[hyphens]{url}
\usepackage[nolist]{acronym} 
\usepackage{pgfplots}
\pgfplotsset{
/pgfplots/short line/.style={
    legend image code/.code={
        \draw[mark repeat=2,mark phase=2,#1] 
            plot coordinates {(0cm,0cm) (0.1275cm,0cm) (0.245cm,0cm)};
    },
},
}

\usetikzlibrary{arrows,shapes,graphs,graphs.standard,quotes,arrows.meta,decorations.markings,positioning}
\usetikzlibrary{fit,positioning,hobby}
\usetikzlibrary{decorations.text}
\usetikzlibrary{decorations.pathreplacing}
\usepackage{transparent}

\pgfplotsset{compat=newest}
\usepackage[bookmarks=false]{hyperref}

\usepackage{nicefrac}
\usepackage{siunitx}
\usepackage{amsmath, amsbsy, amssymb, latexsym}
\usepackage{dsfont}
\usepackage{amsthm}
\usepackage{cuted}
\usepackage{mathtools}
\usepackage{calc}

\usepackage{nicematrix}
\usepackage{booktabs}

\hypersetup{bookmarksdepth=-2}
\usepackage{comment}
\usepackage[utf8]{inputenc}

\usepackage{xcolor}
\usepackage{enumitem}
\usepackage[linesnumbered,titlenumbered,ruled,vlined,resetcount]{algorithm2e}

\SetNlSty{textbf}{\color{black}}{}

\SetCommentSty{mycommentfont}

\SetKwRepeat{Do}{do}{while}
\SetKwProg{Fn}{subroutine}{:}{end}
\let\oldnl\nl%
\newcommand{\nonl}{\renewcommand{\nl}{\let\nl\oldnl}}%
\newcommand{\algrule}[1][.2pt]{\par\vskip.5\baselineskip\hrule height #1\par\vskip.5\baselineskip}

\usepackage{setspace}
\usepackage{algpseudocode}

\makeatletter
\patchcmd\algocf@Vline{\vrule}{\vrule \kern-0.4pt}{}{}
\patchcmd\algocf@Vsline{\vrule}{\vrule \kern-0.4pt}{}{}
\makeatother

\usepackage[normalem]{ulem}
\usepackage{newtxmath}

\usepackage{marginnote}
\tikzset{>=latex}
\usepackage{floatrow}
\SetAlCapNameFnt{\footnotesize}
\SetAlCapFnt{\footnotesize}

\DeclareMathAlphabet{\mathcal}{OMS}{cmsy}{m}{n}
\DeclareMathAlphabet{\mathbb}{U}{msb}{m}{n}

\newcommand{\bigsetxor}{\mathbin{
    \begin{tikzpicture}[baseline=-0.3pt]
    \node[isosceles triangle,
	isosceles triangle apex angle=52, rotate=90, draw, line width=0.4pt, inner sep=0pt, minimum size=0.35cm] {};
    \end{tikzpicture}
}}
\DeclareMathOperator{\setxor}{\triangle}

\DeclareMathOperator{\w}{\textit{w}_\mathrm{H}}
\DeclareMathOperator{\wmin}{\textit{w}_\mathrm{min}}
\DeclareMathOperator{\Awmin}{\textit{A}_{\textit{w}_\mathrm{min}}}

\DeclareMathOperator{\dmin}{\textit{d}_\mathrm{min}}
\DeclareMathOperator{\Admin}{\textit{A}_{\textit{d}_\mathrm{min}}}
\DeclareMathOperator{\C}{\mathcal{P}(\mathcal{I}, \boldsymbol{T})}
\DeclareMathOperator{\Ci}{\mathcal{P}_\textit{i}(\mathcal{I}, \boldsymbol{T})}

\DeclareMathOperator{\Qi}{\mathcal{Q}_\textit{i}(\mathcal{I}, \boldsymbol{T})}
\DeclareMathOperator{\Qiwmin}{\mathcal{Q}_{\textit{i},\wmin}(\mathcal{I}, \boldsymbol{T})}
\DeclareMathOperator{\RM}{\mathcal{RM}}

\newcommand{\ZZ}{\mathbb{Z}}

\newcommand{\ie}{i.\,e.,\ }

\definecolor{mittelblau}{RGB}{0, 126, 198}
\definecolor{violettblau}{cmyk}{0.9, 0.6, 0, 0}
\definecolor{rot}{RGB}{238, 28 35}
\definecolor{apfelgruen}{RGB}{140, 198, 62}
\definecolor{gelb}{RGB}{255, 229, 0}
\definecolor{orange}{RGB}{244, 111, 33}
\definecolor{pink}{RGB}{237, 0, 140}
\definecolor{lila}{RGB}{128, 10, 145}
\definecolor{hellgrau}{RGB}{224, 224, 224}
\definecolor{mittelgrau}{RGB}{128, 128, 128}
\definecolor{dunkelgrau}{RGB}{80,80,80}
\definecolor{anthrazit}{RGB}{19, 31, 31}
\definecolor{darkgreen}{RGB}{34,139,34}
\definecolor{aqua}{RGB}{0, 255, 255}

\definecolor{neuesgruen}{RGB}{61, 173, 65}
\definecolor{dunklereshellgrau}{RGB}{176, 176, 176}
\definecolor{lightgray}{RGB}{211,211,211}
\definecolor{neuesgelb}{RGB}{255,160,0}
\definecolor{neuescyan}{RGB}{69,185,224}
\definecolor{neuesmagenta}{RGB}{197,67,143}
\definecolor{grape}{HTML}{6945C2}
\definecolor{tyrianpurple}{HTML}{6D7BFF}
\definecolor{navyblue}{HTML}{172370}
\definecolor{bluegrotto}{HTML}{009277}
\definecolor{redB}{HTML}{CE054A}

\begin{acronym}
\acro{ML}{maximum likelihood}
\acro{BP}{belief propagation}
\acro{BPL}{belief propagation list}
\acro{LDPC}{low-density parity-check}
\acro{BER}{bit error rate}
\acro{SNR}{signal-to-noise-ratio}
\acro{BPSK}{binary phase shift keying}
\acro{AWGN}{additive white Gaussian noise}
\acro{LLR}{Log-likelihood ratio}
\acro{MAP}{maximum a posteriori}
\acro{FER}{frame error rate}
\acro{BLER}{block error rate}
\acro{SCL}{successive cancellation list}
\acro{SC}{successive cancellation}
\acro{BI-DMC}{Binary Input Discrete Memoryless Channel}
\acro{CRC}{cyclic redundancy check}
\acro{CA-SCL}{CRC-aided successive cancellation list}
\acro{PAC}{polarization-adjusted convolutional}
\acro{BEC}{Binary Erasure Channel}
\acro{BSC}{Binary Symmetric Channel}
\acro{BCH}{Bose-Chaudhuri-Hocquenghem}
\acro{RM}{Reed--Muller}
\acro{RS}{Reed-Solomon}
\acro{SISO}{soft-in/soft-out}
\acro{3GPP}{3rd Generation Partnership Project }
\acro{eMBB}{enhanced Mobile Broadband}
\acro{CN}{check node}
\acro{VN}{variable node}
\acro{GenAlg}{Genetic Algorithm}
\acro{CSI}{Channel State Information}
\acro{OSD}{ordered statistic decoding}
\acro{MWPC-BP}{minimum-weight parity-check BP}
\acro{FFG}{Forney-style factor graph}
\acro{MBBP}{multiple-bases belief propagation}
\acro{URLLC}{ultra-reliable low-latency communications}
\acro{DMC}{discrete memoryless channel}
\acro{SGD}{stochastic gradient descent}
\acro{QC}{quasi-cyclic}
\acro{5G}{fifth generation mobile telecommunication}
\acro{SCAN}{soft cancellation}
\acro{LSB}{least significant bit}
\acro{MSB}{most significant bit}
\acro{AED}{automorphism ensemble decoding}
\acro{AE-SC}{automorphism ensemble successive cancellation}
\acro{PPV}{Polyanskyi-Poor-Verd\'{u}}
\acro{RREF}{reduced row echelon form}
\acro{PTPC}{pre-transformed polar code}
\acro{SCAL}{successive cancellation automorphism list}
\acro{DE}{density evolution}
\acro{PPV}{Polyanskyi-Poor-Verd\'{u}}
\acro{DBT}{dichotomous binary tree}
\acro{PDBT}{perfect dichotomous binary tree}
 
\acro{IBE}{information bit extraction}
\acro{DR}{delta recovery}
\acro{FSSCL}[Fast-SSCL]{fast simplified successice cancellation list}
\acro{HD}{hard decision}
\acro{PAR}{placement and routing}
\acro{PFT}{polar factor tree}
\acro{PM}{path metric}
\acro{PVT}{process, voltage and temperature}
\acro{REP}{repetition}
\acro{SPC}{single parity check}
\end{acronym}

\newtheorem{theorem}{Theorem}
\newtheorem{definition}{Definition}
\newtheorem{corollary}{Corollary}
\newtheorem{proposition}{Proposition}
\newtheorem{lemma}{Lemma}
\newtheorem{example}{Example}
\AtEndEnvironment{example}{\popQED\endexample}

\newcommand\ourTitle{Enumeration of Minimum Weight Codewords of Pre-Transformed Polar Codes by Tree Intersection}

\newif\ifPagebreaks

\begin{document}
\begin{NoHyper}
\title{\ourTitle}

\author{\IEEEauthorblockN{Andreas Zunker, Marvin Geiselhart and Stephan ten Brink}
	\IEEEauthorblockA{
		Institute of Telecommunications, Pfaffenwaldring 47, University of  Stuttgart, 70569 Stuttgart, Germany 
		\\\{zunker,geiselhart,tenbrink\}@inue.uni-stuttgart.de\\
	}
    \thanks{This work is supported by the German Federal Ministry of Education and Research (BMBF) within the project Open6GHub (grant no. 16KISK019).}
}

\maketitle

\begin{abstract}
    \Acp{PTPC} form a class of codes that perform close to the finite-length capacity bounds.
The minimum distance and the number of minimum weight codewords are two decisive properties for their performance.
In this work, we propose an efficient algorithm for determining the number of minimum weight codewords of general \acp{PTPC} that eliminates all redundant visits to nodes of the search tree, thus reducing the computational complexity typically by several orders of magnitude compared to state-of-the-art algorithms.
This reduction in complexity allows, for the first time, the minimum distance properties to be directly considered in the code design of \acp{PTPC}.
The algorithm is demonstrated for randomly pre-transformed \ac{RM} codes and \ac{PAC} codes. 
Furthermore, we design optimal polynomials for \ac{PAC} codes with this algorithm, minimizing the number of minimum weight codewords.

\end{abstract}
\acresetall

\section{Introduction}

While asymptotically achieving the channel capacity  \cite{ArikanMain}, the performance of polar codes in the finite block length regime remains limited by their suboptimal achievable weight spectrum, i.e., the minimum distance $\dmin$ and the large number of minimum weight codewords $A_{\dmin}$, which are the main source of non-correctable errors.
A key advancement to improve the distance spectrum is the removal of low-weight codewords through the introduction of precoding. %
The resulting code is a polar subcode \cite{subCodes} and will be referred to as a \ac{PTPC} hereafter.
Proper precoder design results in very powerful coding schemes that can be decoded by \ac{SCL} with little additional complexity overhead. 
Several precoders have been proposed in the past:
In \cite{talvardyList}, \ac{CRC} codes have been suggested as precoders to prevent incorrect decisions of the \ac{SCL} decoder. %
Similarly, the authors of \cite{subCodes} %
propose parity checks that can be seen as dynamic frozen bits.
Outer convolutional codes have been proposed in \cite{arikan2019pac}, resulting in so-called \ac{PAC} codes that show performance close to the finite block length bounds. %

These advancements went hand in hand with the development of algorithms for counting minimum weight codewords of polar codes and \acp{PTPC}. 
The statistical method proposed in  \cite{Li2021Weightspectrum} investigates the improvement of the distance spectrum by ensembles of random precoders. This method is extended to an asymptotic analysis in \cite{li2023weightspecturmimprovement}.
In contrast, deterministic/exact methods analyze a concrete realization of a plain or pre-transformed polar code.
For plain polar codes, \cite{bardet_polar_automorphism} proposes to use symmetries (permutations) of polar codes to generate all minimum weight codewords from minimum weight generators. While of low complexity, the algorithm is limited to polar codes following the partial order.
This method is extended to obtain the full weight distribution of polar codes up to length 128 \cite{Yao2021polarweightdistribution}.
\cite{Rowshan2023Impact} proposes to count minimum weight codewords by finding combinations of rows that generate them.
Although originally designed for plain polar codes and related to \cite{bardet_polar_automorphism} as shown in \cite{RowshanFormation}, this method is also applicable to \acp{PTPC} such as \ac{PAC} codes. 
While more complex, algorithms for more terms of the weight distribution are proposed in \cite{Miloslavskaya2022partialweightdistribution} and \cite{PartialEnumPAC}.
Due to the unfavorable scaling of these algorithms, however, it was not yet feasible to compute the number of minimum weight codewords for long polar codes \cite{li2023weightspecturmimprovement}.
Our contributions are:
\begin{itemize}
    \item We propose an improved algorithm for counting the minimum weight codewords of \acp{PTPC}.
    Compared to \cite{Rowshan2023Impact}, the proposed algorithm eliminates all redundant visits to each coset in the search tree and is therefore significantly less complex.
    Moreover, the proposed algorithm supports arbitrary upper-triangular precoding rather than only convolutional precoding.

    \item We find bounds on the minimum distance $\dmin$ and the number of minimum weight codewords $\Admin$ of a \ac{PTPC} based on its rate-profile.
    
    \item We demonstrate our low-complexity weight-enumeration algorithm by providing numerical results for very long codes (such as pre-transformed $\RM(10,21)$ codes with ${N=2097152}$), for which only probabilistic estimates existed \cite{li2023weightspecturmimprovement}, and list optimal polynomials for \ac{PAC} codes.
\end{itemize}

\section{Preliminaries}

\subsection{Notations}
The set of integers between $i$ and $j-1$ is denoted by $\mathbb{Z}_{i:j}$, with the shorthand notation $\mathbb{Z}_{j}=\mathbb{Z}_{0:j}$.
An integer $i \in \mathbb Z_{2^n}$ is implicitly represented by its $n$-bit \ac{MSB}-first binary expansion ${i = \sum_{l=0}^{n-1} i_{(l)}\cdot 2^l \mathrel{\widehat{=}} i_{(n-1)} \dots i_{(1)} i_{(0)}}$, with support $\mathcal{S}_i=\operatorname{supp}(i)$,
complementary support set $\mathcal{S}_i^\mathrm{c}=\mathbb{Z}_n\setminus\mathcal{S}_i$ and Hamming weight $\w(i)=\left|\mathcal{S}_i\right|$. 
Operations \textit{NOT}, \textit{AND}, and \textit{OR} are applied bitwise and denoted by $\bar i$, $i \land j$, and $i\lor j$, respectively.
The symmetric difference of two sets $\mathcal{A}$ and $\mathcal{B}$ is denoted as $\mathcal{A} \setxor \mathcal{B} = \left(\mathcal{A} \setminus \mathcal{B}\right) \cup \left(\mathcal{B} \setminus \mathcal{A}\right)$.

\subsection{Error Probability of Binary Linear Block Codes}
Let $\w(\boldsymbol c)$ denote the Hamming weight of a binary vector $\boldsymbol c \in \mathbb{F}_2^N$.
The minimum distance a of binary linear block code $\mathcal C(N,K)$ with code length $N$ and code dimension $K$ is
\begin{equation*}
    \dmin(\mathcal{C}(N,K)) = \min\limits_{\boldsymbol{c},\boldsymbol{c}' \in \mathcal {C},\, \boldsymbol{c} \neq \boldsymbol {c}'} \w(\boldsymbol{c} \oplus \boldsymbol{c}') = \min\limits_{\boldsymbol {c} \in \mathcal{C},\, \boldsymbol{c} \neq \boldsymbol{0}} \w(\boldsymbol{c}).
\end{equation*}
The number of codewords $\boldsymbol{c} \in \mathcal{C}(N,K)$ with Hamming weight $w$ is given by
\begin{equation*}
    A_w\left(\mathcal{C}(N,K)\right) = \left| \left\{\boldsymbol{c} \in \mathcal{C}(N,K) \,\middle|\, \w(\boldsymbol{c}) = w \right\} \right|  %
\end{equation*} 
and the ordered set $\mathcal{A}_{\mathcal{C}} = \left\{A_0,\,A_1,\,\dots,\,A_N\right\}$ is called the weight spectrum of code $\mathcal{C}(N,K)$. 
For \ac{BPSK} mapping and transmission over an \ac{AWGN} channel with \ac{SNR} $E_\mathrm{b}/N_0$, the \ac{FER} under \ac{ML} decoding $P^\mathrm{ML}_\mathrm{FE}$ of a binary linear block code $\mathcal C(N,K)$ with code rate $R = K/N$ can be upper bounded by the union bound \cite{UnionBound} as 
\begin{equation} %
    P^\mathrm{ML}_\mathrm{FE} \leq P^\mathrm{UB}_\mathrm{FE} = \sum\limits_w A_w \cdot Q\left(\sqrt{2 w \cdot R \cdot \frac{E_\mathrm{b}}{N_0}}\right), \label{eq:unionbound}
\end{equation}
where $Q(\cdot)$ is the complementary cumulative distribution function of the standard normal distribution.
As the \ac{SNR} increases, the contribution of codewords with high weight $w$ in the computation of $P^\mathrm{UB}_\mathrm{FE}$ vanishes. 
Hence, (\ref{eq:unionbound}) is well approximated by the first few non-zero terms (\emph{truncated union bound}).
Thus, the minimum distance $\dmin$ and the number of minimum-weight codewords $A_{\dmin}$ are two decisive properties for the achievable \ac{ML} performance of a code.

\subsection{Polar Codes}
Polar codes are constructed on the basis of the $n$-fold application of the polar transform, resulting in polarized synthetic bit channels \cite{ArikanMain}. Reliable bit channels, denoted by the information set $\mathcal{I}$, are used for information transmission, while the remaining unreliable bit channels $\mathcal{F}=\ZZ_N \setminus \mathcal{I}$ (frozen set) are set to zero. 
The generator matrix $\boldsymbol G = (\boldsymbol G_{N})_{\mathcal I}$ of an $(N,K)$ polar code  $\mathcal P(\mathcal I)$ is formed by those $K$ rows of the matrix $\boldsymbol G_N = \begin{bsmallmatrix} 1 & 0 \\ 1 & 1 \end{bsmallmatrix}^{\otimes n}$ indicated by $\mathcal{I}$, where $(\cdot)^{\otimes n}$ denotes the $n$-th Kronecker power and $N = 2^n$.
Let $\boldsymbol m = \boldsymbol u_\mathcal{I} \in \mathbb F_2^{K}$ denote the message and  $\boldsymbol u_\mathcal{F} = \boldsymbol 0$ the frozen bits. Then, polar encoding can be written as $\boldsymbol c = \boldsymbol m \cdot \boldsymbol G = \boldsymbol u \cdot \boldsymbol G_N$. %

Although the optimal code design $\mathcal I$ depends on the transmission channel, some synthetic bit channels are always more reliable than others \cite{bardet_polar_automorphism}. 
Thus, the bit channels exhibit a \textit{partial order}, in which $i \preccurlyeq j$ denotes that channel $j$ is at least as reliable as channel $i$, defined by:
\begin{enumerate}
  \item \textit{Left swap}: If $i_{(l)} > j_{(l)}$, and $i_{(l+1)} < j_{(l+1)}$ for ${l\in \mathbb Z_{n-1}}$, while all other $i_{(l')} = j_{(l')}$, $l' \notin \{l,l+1\}$, then $i \preccurlyeq j$. %
  \item \textit{Binary domination}: If $i_{(l)}\le j_{(l)}$ for all $l\in \ZZ_n$, then $i \preccurlyeq j$. 
\end{enumerate}
All remaining relations are derived from transitivity.
A polar code $\mathcal P(\mathcal I)$ is called a \textit{decreasing monomial code} and said to comply with the partial order ``$\preccurlyeq$'' if and only if
$\mathcal I$ fulfills the partial order as $\forall i \in \mathcal I,\;\forall j \in \mathbb Z_N$ with $i \preccurlyeq j \implies j \in \mathcal I$.

\ac{RM} codes are a special case of decreasing monomial codes.
The \ac{RM} code of order $r$ and length ${N=2^n}$, denoted by $\RM(r,n)$, is exactly the plain polar code $\mathcal{P}(\mathcal{I})$ with information set ${\mathcal{I} = \left\{i \in \mathbb{Z}_N \,\middle|\,\w(i) \geq n-r\right\}}$. An $\RM(r,n)$ code has minimum distance ${d_\mathrm{min} = 2^{n-r}}$.

\subsection{Pre-Transformed Polar Codes (PTPCs)}
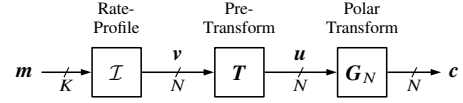
\begin{figure}[htp]
    \vspace{-0.2cm}
	\centering
	\resizebox{.7\columnwidth}{!}{\colorlet{dF}{orange!70!white}%
\begin{tikzpicture}[
    start chain = 1 going below,    node distance = 0pt,
    cell/.style={draw, fill=white, minimum width=1em, minimum height=1em, inner sep=0pt,
                outer sep=0pt, on chain},
    long/.style={draw, fill=white, minimum width=1em, minimum height=2em,  inner sep=0pt,
                outer sep=0pt, on chain},
    tick/.pic = {
    		\draw[line width=0.5pt] (-0.4mm,-0.8mm) -- (0.4mm,0.8mm);
    	}
  ]
  
\node[] (m) at (0.5, 0) {$\boldsymbol m$};
\node[dspsquare] (rp) at (2, 0) {$\mathcal I$};
\node[align=center, anchor=north, above = 0.1cm of rp]  {\footnotesize Rate-\\[-3pt]\footnotesize Profile};
\node[dspsquare] (pt) at (4, 0) {$\boldsymbol T$};
\node[align=center, anchor=north, above = 0.1cm of pt]  {\footnotesize Pre-\\[-3pt]\footnotesize Transform};
\node[dspsquare] (gn) at (6, 0) {$\boldsymbol G_N$};
\node[align=center, anchor=north, above = 0.1cm of gn]  {\footnotesize Polar\\[-3pt]\footnotesize Transform};

\node[] (c) at (7.5, 0) {$\boldsymbol c$};

\draw[dspconn](m)  -- node[] (tickM) {} node[below] {\footnotesize $K$} (rp);
\draw[dspconn](rp) -- node[above,yshift=2pt] {$\boldsymbol v$} node[] (tickV) {} node[below] {\footnotesize $N$}(pt);
\draw[dspconn](pt) -- node[above,yshift=2pt] {$\boldsymbol u$} node[] (tickU) {} node[below] {\footnotesize $N$}(gn);
\draw[dspconn](gn) -- node[] (tickC) {} node[below] {\footnotesize $N$} (c);

\pic at (tickM) {tick};
\pic at (tickV) {tick};
\pic at (tickU) {tick};
\pic at (tickC) {tick};

\end{tikzpicture}}
	\caption{\footnotesize Block diagram of pre-transformed polar encoding. %
   }
	\label{fig:Encoder}
\end{figure}

A polar code is \emph{pre-transformed}, if the message $\boldsymbol{m}$ is transformed by a (linear) mapping before the polar transform~$\boldsymbol G_N$. 
Fig.~\ref{fig:Encoder} shows the encoder block diagram of a \ac{PTPC} $\mathcal{P}(\mathcal{I},\boldsymbol{T})$.
In the context of \acp{PTPC}, $\mathcal I$ is called the \emph{rate-profile}.
The pre-transformation is sandwiched between the insertion of the frozen bits (rate-profiling) and the polar transform \cite{arikan2019pac}. It can be described by a multiplication with an upper-triangular pre-transformation matrix $\boldsymbol T \in \mathbb F_2^{N \times N}$, \ie ${t_{i,j}=0}$ if $i > j$ and ${t_{i,i}=1}$ for $i \in \mathcal I$.
From the rate-profiled message $\boldsymbol{v}$, where $\boldsymbol{v}_{\mathcal{I}} = \boldsymbol{m} \in \mathbb{F}_2^K$ and $\boldsymbol{v}_\mathcal{F} = \boldsymbol{0}$,
the pre-transformed message is found as $\boldsymbol{u} = \boldsymbol{m} \cdot \boldsymbol{T}_\mathcal{I}  = \boldsymbol v \cdot \boldsymbol T$.
Finally, the codewords $\boldsymbol c = \boldsymbol u \cdot \boldsymbol G_N$ are obtained by applying the polar transform.
In contrast to plain polar codes, the pre-transformation introduces dynamic frozen bits $f \in \mathcal F$ whose values are dependent on the previous message bits as $u_f = \bigoplus_{h=0}^{f-1} v_h \cdot t_{h,f}$.

Note that this perspective on \acp{PTPC} differs from conventional code concatenation, where both inner and outer code have a rate $R < 1$. In contrast, \acp{PTPC} as used in this paper add redundancy solely using the rate-profile $\mathcal{I}$, while the following transformations $\boldsymbol{T}$ and $\boldsymbol{G}_N$ are rate-1. 
This way, many different variants of polar codes (such as plain and \ac{CRC}-aided polar codes, \ac{PAC} codes, etc.) can be described under the umbrella of \acp{PTPC}\footnote{Any binary linear code can be described as an equivalent \ac{PTPC} \cite{subCodes}.}. 
For plain polar codes, simply $\boldsymbol T = \boldsymbol I_{N}$.
In the case of \ac{PAC} codes, the pre-transformation is given by the convolution of $\boldsymbol{v}$ with a generator polynomial $p(x) = \sum_{h=0}^q p_h \cdot x^h$
of degree $\operatorname{deg}(p(x))=q$.

\section{Counting Minimum Weight Codewords of Pre-Transformed Polar Codes}\label{sec:counting}

Our goal is to find the minimum distance $\dmin$ and the corresponding number of $\dmin$-weight codewords $\Admin$ of a \ac{PTPC} $\C$.
For this, we will exploit the underlying structure of polar codes. 
Thus, our method is best suited for \acp{PTPC} $\C$ with polar-like rate-profile~$\mathcal I$ and upper-triangular pre-transformation matrix~$\boldsymbol{T}$.
To simplify the problem, the general approach is to split the code $\C = \{\boldsymbol{0}\} \cup \bigcup_{i \in \mathcal I} \Ci$ into $K=|\mathcal I|$ distinct \emph{cosets} ${\Ci = \left\{\boldsymbol{u} \cdot \boldsymbol{G}_N \,\middle|\, \boldsymbol{u} \in \Qi\right\}}$ which 
are obtained by the one-to-one polar transform mapping $\boldsymbol{G}_N$ of the pre-transformed messages
\begin{equation*}%
    \Qi = \Big\{\boldsymbol{t}_i \oplus \bigoplus_{h \in \mathcal{H}} \boldsymbol{t}_h\,\Big|\,\mathcal{H} \subseteq \mathcal{I} \setminus \mathbb{Z}_{i+1} \Big\}.
\end{equation*}
In other words, $\Ci$ contains all codewords whose corresponding messages $\Qi$ have their first non-zero bit in position $i$. 
The row $\boldsymbol{g}_i$ is called the \emph{coset leader} of $\Ci$.
Instead of examining the whole code at once, each coset can be considered individually.
We subsequently denote \emph{plain} (i.e., without pre-transform) and \emph{universal} (i.e., without considering a specific rate-profile) polar cosets by ${\mathcal{P}_i(\mathcal I) = \mathcal{P}_i(\mathcal{I},\boldsymbol{T}=\boldsymbol{I}_N)}$ and $\mathcal{P}_i=\mathcal{P}_i(\mathcal I=\mathbb{Z}_N)$, respectively.
Their corresponding messages are denoted by $\mathcal{Q}_i(\mathcal{I}) = \mathcal{Q}_i(\mathcal{I},\boldsymbol{I}_N)$ and $\mathcal{Q}_i=\mathcal{Q}_i(\mathbb{Z}_N)$.

\subsection{Minimum Distance of \acp{PTPC}}
To find the number of $\dmin$-weight codewords, the minimum distance must be considered first.
For this, we review a lower bound on the minimum distance $\dmin(\C)$ of a \ac{PTPC}.
As shown in \cite{RowshanFormation}, any codeword in a universal coset $\mathcal P_i$ has at least the Hamming weight of the coset leader $\boldsymbol{g}_i$, i.e.,
\begin{equation}\label{eq:UTPdmin}
    \w\Big(\boldsymbol{g}_i \oplus \bigoplus\limits_{h \in \mathcal{H}} \boldsymbol{g}_h\Big) \geq \w(\boldsymbol{g}_i),
\end{equation}
for all $\mathcal H \subseteq \mathbb Z_{i+1:N}$.
From $\mathcal{P}_i(\mathcal{I}) \subseteq \mathcal{P}_i$ and $\boldsymbol{g}_i \in \mathcal{P}_i(\mathcal{I})$ follows that the minimum weight of the codewords of $\mathcal{P}_i(\mathcal{I})$ is $\w(\boldsymbol{g}_i)$.
The Hamming weight of a row $\boldsymbol g_i$ can be efficiently calculated as ${\w(\boldsymbol g_i)=2^{\w(i)}}$ \cite{RowshanFormation}.
Consequently, the minimum distance of a plain polar code $\mathcal{P}(\mathcal{I})$ is found as 
\begin{equation*}
    \dmin(\mathcal{P}(\mathcal{I}))=\min_{i \in \mathcal{I}}\left\{\w(\boldsymbol{g}_i)\right\}=\min_{i \in \mathcal{I}}\big\{2^{\w(i)}\big\}.
\end{equation*}
For an upper-triangular pre-transformation matrix $\boldsymbol{T}$ follows $\Ci \subseteq \mathcal{P}_i$.
Therefore, we find with (\ref{eq:UTPdmin}) that the codewords of $\Ci$ have at least the Hamming weight $\w(\boldsymbol{g}_i)$ of the coset leader $\boldsymbol{g}_i$.
Hence, a lower bound on the minimum distance of \acp{PTPC} can be derived as
\begin{equation}\label{eq:dmin lower bound}
    \dmin(\C) \geq \dmin(\mathcal{P}(\mathcal{I})) = \wmin.
\end{equation}
In the following, the lowest potential codeword weight $\wmin$ is a constant dependent only on the (fixed) rate-profile $\mathcal{I}$ of the given \ac{PTPC} $\C$.

\subsection{Minimum Weight Codewords of \acp{PTPC}}
Following the lower bound in (\ref{eq:dmin lower bound}), two cases occur for the minimum distance $\dmin$ of a \ac{PTPC} $\C$:
\begin{enumerate}
    \item $\dmin(\C) = \wmin$: \newline
    The minimum distance is equal to the lower bound and thus $\Awmin(\C) > 0$.
    We consider this case and propose an enumeration method that counts the codewords with weight $\wmin$ of a \ac{PTPC} $\C$.
    \item $\dmin(\C) > \wmin$: \newline
    The minimum distance is greater than the lower bound, as all $\wmin$-weight codewords have been eliminated by the pre-transform $\boldsymbol{T}$.
    In this case, the proposed algorithm will just find $\Awmin(\C) = 0$.
    Consequently, more general methods as introduced in \cite{Miloslavskaya2022partialweightdistribution} or \cite{PartialEnumPAC} have to be applied to find $\dmin$ and $\Admin$ of $\C$.
\end{enumerate}

Our method is based on the construction of $\wmin$-weight codewords of polar codes introduced in \cite{RowshanFormation}.
Following this formalism, the enumeration algorithm proposed in \cite{Rowshan2023Impact} counts the codewords of weight $\wmin$ of \ac{PAC} codes.
We will generalize this method to arbitrary \acp{PTPC} and reduce the complexity by eliminating all redundant operations.
From ${\Ci \subseteq \mathcal{P}_i}$ and (\ref{eq:UTPdmin}) follows that only cosets with coset index 
\begin{equation*}
    i \in \mathcal{I}_{\wmin} = \left\{j \in \mathcal{I} \,\middle|\, \w(\boldsymbol{g}_j) = \wmin\right\}
\end{equation*}
can contain codewords of weight $\wmin$.
Hence, only these cosets or the respective message sets which generate them have to be considered for the computation of $A_{\wmin}(\C)$.
The set of (pre-transformed) messages that form the $\wmin$-weight codewords $\mathcal{P}_{i,\wmin}(\mathcal{I},\boldsymbol{T})$ of a coset $\Ci$ is defined as
\begin{equation*}
    \mathcal{Q}_{i,\wmin}(\mathcal{I},\boldsymbol{T}) = \left\{\boldsymbol{u} \in \Qi \,\middle|\, \w(\boldsymbol{u} \cdot \boldsymbol{G}_N) = \wmin \right\}.
\end{equation*}
Consequently, 
the number of $\wmin$-weight codewords of a \ac{PTPC} can be found as 
\begin{equation}\label{eq:Awmin}
    A_{\wmin}(\C) = \sum_{i \in \mathcal{I}_{\wmin}} \left|\mathcal{Q}_{i,\wmin}(\mathcal{I},\boldsymbol{T})\right|.
\end{equation}
The formalism introduced in \cite{RowshanFormation} allows to find all messages ${\mathcal{Q}_{i,\wmin} = \mathcal{Q}_{i,\wmin}(\mathbb{Z}_N,\boldsymbol{I}_N)}$ that form the codewords with weight $\wmin=\w(\boldsymbol{g}_i)$ of a universal polar coset $\mathcal{P}_i$.
Since both $\mathcal{Q}_{i,\wmin} \subseteq \mathcal{Q}_i$ and $\Qi \subseteq \mathcal{Q}_i$, we can find the messages that generate the $\wmin$-weight codewords of a polar coset $\Ci$ by the intersection of the two considered message sets as
\begin{equation*}
    \mathcal{Q}_{i,\wmin}(\mathcal{I},\boldsymbol{T}) =  \Qi \cap \mathcal{Q}_{i,\wmin}.
\end{equation*}
The set of ``1'' bit positions $\operatorname{supp}(\boldsymbol{u})$ of a message $\boldsymbol{u}$ directly corresponds to the row combination $\boldsymbol{u} \cdot \boldsymbol{G}_N = \bigoplus_{h \in \operatorname{supp}(\boldsymbol{u})} \boldsymbol{g}_h$.
The construction of codewords of a universal polar coset $\mathcal{P}_i$ with weight 
\begin{equation*}
    \w(\boldsymbol{u} \cdot \boldsymbol{G}_N) = \w\Big( \underbrace{\boldsymbol{g}_i}_{\substack{\text{coset}\\\text{leader}}} \oplus \underbrace{\bigoplus\limits_{j \in \mathcal{J}} \boldsymbol{g}_j}_\text{core rows} \oplus \underbrace{\bigoplus\limits_{m \in \mathcal{M}_i(\mathcal{J})} \boldsymbol{g}_m}_\text{balancing rows}\Big) = \w(\boldsymbol{g}_i)
\end{equation*}
can be formulated by the combination of the coset leader $\boldsymbol{g}_i$ with a set of \emph{core rows} $\mathcal{J} \subseteq \mathcal{K}_i$ and dependent \emph{balancing rows} $\mathcal{M}_i(\mathcal{J}) \subseteq \mathcal{K}_i^\mathrm{c}$.
Thus, the set of ``1'' bit positions of a message $\boldsymbol{u}$ that forms a codeword of weight $\wmin=\w(\boldsymbol{g}_i)$ of $\mathcal{P}_i$ is 
\begin{equation}\label{eq:Ui set}
    \operatorname{supp}(\boldsymbol{u}(i,\mathcal{J})) = \mathcal{U}_i(\mathcal{J}) = \{i\} \cup \mathcal{J} \cup \mathcal{M}_i(\mathcal{J}).
\end{equation}
Hence, the set of all messages $\boldsymbol{u}$ that form the $\wmin$-weight codewords of a universal polar coset $\mathcal{P}_i$ is 
\begin{equation*}%
        \mathcal{Q}_{i,\wmin} = \mathcal{Q}_{i,\wmin}(\mathbb{Z}_N,\boldsymbol{I}_N) = \left\{\boldsymbol{u}(i,\mathcal{J})\,\middle|\,\mathcal{J} \subseteq \mathcal{K}_i\right\}.
\end{equation*}
The core rows $\mathcal{J} \subseteq \mathcal{K}_i$ are found as subsets of the (subsequent) rows that individually do not increase the weight of the coset leader $\boldsymbol g_i$ as
\begin{equation}\label{eq:ki}
    \mathcal K_i = \left\{ j \in \mathbb Z_{i+1:N} \,\middle|\, \w(\boldsymbol g_i \oplus \boldsymbol g_j) = \w(\boldsymbol g_i) \right\}.
\end{equation}
The corresponding complementary set is $\mathcal K_i^\mathrm{c} = \mathbb Z_{i+1:N} \setminus \mathcal K_i$.
As shown in \cite{RowshanFormation}, the Hamming weight of the sum of two rows can be computed as 
\begin{equation}\label{eq:merged row weight}
    \w(\boldsymbol g_i \oplus \boldsymbol g_j) = 2^{\w(i)} + 2^{\w(j)} - 2^{\w(i \land j)+1}.
\end{equation}
For $i < j$, $\w(j)-\w(i~\land~j)=\w(\bar i~\land~j)$ is fulfilled.
Thus, with (\ref{eq:merged row weight}) we find for $i < j$ that
\begin{equation}
    \label{eq:cond increase weight}
    \begin{split}
        j \in \mathcal K_i \Leftrightarrow \w(\boldsymbol g_i \oplus \boldsymbol g_j) = \w(\boldsymbol g_i) &\Leftrightarrow \w(\bar i \land j) = 1,\\
        j \in \mathcal K_i^\mathrm{c} \Leftrightarrow \w(\boldsymbol g_i \oplus \boldsymbol g_j) > \w(\boldsymbol g_i) &\Leftrightarrow \w(\bar i \land j) > 1,        
    \end{split}
\end{equation}
efficiently checking whether $j \in \mathcal{K}_i$ or $j \in \mathcal{K}_i^\mathrm{c}$.
As each core row-combination $\mathcal{J} \subseteq \mathcal{K}_i$ forms one $\wmin$-weight codeword, the number of these codewords in $\mathcal {P}_i$ is %
\begin{equation}\label{eq:num wmin Qiwmin}
    \left|\mathcal{Q}_{i,\wmin}\right| =  \sum_{j=0}^{\left|\mathcal{K}_i\right|} \binom{\left|\mathcal{K}_i\right|}{j} = 2^{\left|\mathcal{K}_i\right|}.
\end{equation}

In addition to a set of core rows $\mathcal J \subseteq \mathcal{K}_i^\mathrm{c}$, the balancing rows ${\mathcal M_i(\mathcal J) \subseteq \mathcal K_i^\mathrm{c}}$ are required to obtain a $\wmin$-weight codeword.
Following \cite[Algorithm~2]{Rowshan2023Impact}, the balancing row set $\mathcal M_i(\mathcal J)$ can be computed iteratively, starting with $\mathcal M_i(\emptyset) = \emptyset$, as
\begin{equation}\label{eq:Mset}
    \mathcal{M}_i\left(\tilde{\mathcal{J}} \cup \{j\}\right) = \mathcal{M}_i\left(\tilde{\mathcal{J}}\right) \setxor \tilde{\mathcal{M}}_i\left(j,\tilde{\mathcal{J}}\right),
\end{equation}
where $j \in \mathcal{J} \setminus \tilde{\mathcal{J}}$, $\tilde{\mathcal{J}} \subset \mathcal{J}$, and
\begin{equation}\label{eq:add to Mset}
    \tilde{\mathcal{M}}_i(j,\mathcal {J}) = \underset{{\substack{k \in \mathcal{J} \cup \mathcal{M}_i(\mathcal{J}),\\\bar{i} \land j \land k = 0}}}{\bigsetxor} \left\{\mu_i(j,k)\right\},
\end{equation}
with\footnote{The direct translation from \cite{Rowshan2023Impact} into our notation would be ${\mu_i(j,k) = (\bar i \land (j \lor k)) \lor (i \land j \land k)}$, but since ${\bar i \land j \land k = 0}$ implies that ${i \land j \land k = j \land k}$, we can simplify the expression.} ${\mu_i(j,k) = (\bar i \land (j \lor k)) \lor (j \land k)}$.
Using this formalism, we can compute all messages $\mathcal{Q}_{i,\wmin}$ that generate a codeword with weight $\wmin=\w(\boldsymbol{g}_i)$ of a universal polar coset $\mathcal{P}_i$.
Thus, the number of $\wmin$-weight codewords of a polar coset $\Ci$ can be found as
\begin{equation}\label{eq:intersection cardinality}
    \left|\mathcal{Q}_{i,\wmin}(\mathcal{I},\boldsymbol{T})\right| = \left|\Qi \cap \mathcal{Q}_{i,\wmin}\right|.
\end{equation}
A na\"ive approach to computing $\left|\mathcal{Q}_{i,\wmin}(\mathcal{I},\boldsymbol{T})\right|$ which involves finding all elements of $\Qi$ or $\mathcal{Q}_{i,\wmin}$ one by one naturally entails a high computational complexity.
In the following, we will propose a low-complexity method to compute (\ref{eq:intersection cardinality}).

\subsection{Message Trees}
Our approach is to represent the message sets $\Qi$ and $\mathcal{Q}_{i,\wmin}$ as \acp{PDBT} and then compute the intersection of the trees.
A \ac{PDBT} is a special \ac{DBT} which is defined as follows.
\begin{definition}[\Acf{DBT}]\label{def:dichotomous binary tree}
    \normalfont
    A \ac{DBT} is a binary tree that has, apart from the root node, two distinct types of levels:
    \begin{enumerate}
        \item \emph{Einzelchild levels}%
        $\mathcal{L}_1$: 
        All nodes do not have a sibling. 
        \item \emph{Sibling levels $\mathcal{L}_2$}: 
        All nodes share their parent node with a sibling node.
    \end{enumerate}   
    The value of each node is either ``0'' or ``1'', where a left child node is always ``0'' and a right child node is always ``1''. If all leaf nodes are at the last level, the tree is called a \ac{PDBT}.
\end{definition}

To prove that both message sets $\Qi$ and $\mathcal{Q}_{i,\wmin}$ can be represented as \acp{PDBT}, we first establish the following lemma.

\begin{lemma}\label{proposition:dichotomous binary tree}
    A set of binary vectors $\mathcal{B} \subseteq \mathbb{F}_2^N$ can be represented by the full paths through a \ac{PDBT} of height 
    $N$ if and only if $b_0=b_{\smash{0}}'$ for all ${\boldsymbol{b},\boldsymbol{b}' \in \mathcal{B}}$ and $\mathcal{B}$ has cardinality $\left|\mathcal{B}\right|=2^{\left|\mathcal{L}^*\right|}$, where
    \begin{equation} \label{eq:sibling levels}
        \mathcal{L}^* = \left\{\min\left\{l \in \mathbb{Z}_{1:N} \,\middle|\, b_l \neq b_{\smash{l}}'\right\}\,\middle|\,\boldsymbol{b},\boldsymbol{b}' \in \mathcal{B},\, \boldsymbol{b} \neq \boldsymbol{b}'\right\}
    \end{equation}  
    is the set of levels that contain sibling nodes.
\end{lemma}
\begin{proof}
    ``$\Rightarrow$'':
    A \ac{PDBT} has sibling layers $\mathcal{L}_2$ and thus $2^{\left|\mathcal{L}_2\right|}$ leaf nodes.
    Consequently, the corresponding set $\mathcal{B}$ of full paths through the tree has cardinality $\left|\mathcal{B}\right|=2^{\left|\mathcal{L}_2\right|}$ and ${\mathcal{L}^*=\mathcal{L}_2}$.

    ``$\Leftarrow$'': 
    Any set $\mathcal{B} \subseteq \mathbb{F}_2^N$, where $b_0=b_{\smash{0}}'$ for all ${\boldsymbol{b},\boldsymbol{b}' \in \mathcal{B}}$, can be represented as a binary tree.
    Therefore, from (\ref{eq:sibling levels}) directly follows ${\left|\mathcal{B}\right| \leq 2^{\left|\mathcal{L}^*\right|}}$.
    For ${\left|\mathcal{B}\right| < 2^{\left|\mathcal{L}^*\right|}}$, at least one node at a sibling level $\mathcal{L}^*$ would have no sibling node. 
    Consequently, the set must have cardinality $\left|\mathcal{B}\right|=2^{\left|\mathcal{L}^*\right|}$ for the corresponding binary tree to be dichotomous.
\end{proof}

\begin{proposition}\label{proposition: Qi tree}
    The message set $\Qi$ corresponding to an arbitrary polar coset $\Ci$ can be represented as a \ac{PDBT} with sibling levels ${\mathcal{L}_2 = \mathcal{I} \setminus \mathbb{Z}_{i+1}}$.
\end{proposition}
\begin{proof}
    The individual bits of a pre-transformed message ${\boldsymbol{u} \in \Qi}$ are found as
    \begin{equation}\label{eq:pretransform naive}
        u_k = v_k \oplus t_{i,k} \oplus \bigoplus_{h=i+1}^{k-1} v_h \cdot t_{h,k},\; 
    \end{equation}
    with $v_i=1$, $v_k \in \mathbb{F}_2$ for $k \in \mathcal{I} \setminus \mathbb{Z}_{i+1}$ and $v_k=0$ otherwise.
    This directly implies ${\mathcal{L}_2 = \mathcal{I} \setminus \mathbb{Z}_{i+1}}$ and $\left|\Qi\right|=2^{\left|\mathcal{I} \setminus \mathbb{Z}_{i+1}\right|}$.
    Thus, by Lemma~\ref{proposition:dichotomous binary tree}, $\Qi$ can be represented as a \ac{PDBT}.
\end{proof}
Following Proposition~\ref{proposition: Qi tree}, a tree-based approach can be used to efficiently compute the messages $\boldsymbol{u} \in \Qi$.
The complexity of finding the pre-transformed messages $\boldsymbol{u}$ from messages $\boldsymbol{v}$, as in (\ref{eq:pretransform naive}), can be reduced by bringing pre-transform $\boldsymbol{T}$ into its systematic form $\boldsymbol{T}_\mathrm{sys}$, where $(\boldsymbol{T}_\mathrm{sys})_\mathcal{I}=\operatorname{RREF}(\boldsymbol{T}_\mathcal{I})$ is in \ac{RREF} and 
$(\boldsymbol{T}_\mathrm{sys})_\mathcal{F}=\boldsymbol{0}$.
Consequently, we find $\boldsymbol{u}_\mathcal{I}=\boldsymbol{v}_\mathcal{I}$.
Thus, the dynamic frozen bits at the einzelchild levels $f\in\mathcal{F} \setminus \mathbb{Z}_{i+1}$ can be computed as $u_f = \bigoplus_{h=i}^{f-1} u_h \cdot t_{h,f}$.

To show that the message set $\mathcal{Q}_{i,\wmin}$ can be characterized as a \ac{PDBT}, we additionally require the subsequent lemma.
\begin{lemma}\label{lemma:M construction}
    Let $i<j<k$. If $\bar i \land j \land k = 0$ is fulfilled, then ${\mu_i(j,k) = (\bar i \land (j \lor k)) \lor (j \land k)} > k$.
\end{lemma}
\begin{proof}
    Let us split the bit positions in the support $\mathcal{S}_i$ of coset index $i$ as and the complementary set $\mathcal{S}_i^\mathrm{c}$.
    Subsequently, $l \in \mathcal S_i$ implies that ${\mu_i(j,k)_{(l)} = j_{(l)} \land k_{(l)}}$ and $l \in \mathcal{S}_i^\mathrm{c}$ implies that ${\mu_i(j,k)_{(l)} = j_{(l)} \lor k_{(l)}}$.
    Let the most significant bit index $l$ where $h_{(l)}>i_{(l)}$ be defined as 
    $l_i^*(h) = \max\left(\mathcal{S}_i^\mathrm{c} \cap \mathcal{S}_h\right)$.
    Consequently, if ${h > i}$, then $h_{(l)}=1$ for all $l \in \mathcal{S}_i \cap \mathbb{Z}_{l_i^*(h):n}$. 
    Thus, $l_i^*(k) > l_i^*(j)$ if $i < j < k$ as ${\bar i \land j \land k = 0}$. 
    We examine the most significant bit positions $\mathbb{Z}_{l_i^*(j):n}$.
    First, for $l \in \mathcal{S}_i \cap \mathbb{Z}_{l_i^*(j):n}$ we find $j_{(l)}=1$ as above and thus $\mu_i(j,k)_{(l)} = k_{(l)}$.
    Second, for ${l \in \mathcal{S}_i^\mathrm{c} \cap \mathbb{Z}_{l_i^*(j):n}}$ we find $j_{(l)}=0$ since ${l > l_i^*(j)}$ and thus $\mu_i(j,k)_{(l)} = k_{(l)}$.
    Finally, for $l = l_i^*(j)$ follows $\mu_i(j,k)_{(l)} > k_{(l)}$ and therefore $\mu_i(j,k) > k > j$.
\end{proof}

\begin{proposition}\label{proposition: Qiwmin tree}
    The message set $\mathcal{Q}_{i,\wmin}$ that forms the codewords with weight $\wmin$ of a universal polar coset $\mathcal P_i$ can be represented as a \ac{PDBT} with sibling levels ${\mathcal{L}_2 = \mathcal{K}_i}$.
\end{proposition}
\begin{proof}
    Let $\mathcal{J} \cup \{j\} \subseteq \mathcal{K}_i$ be a set of core rows obtained from $\mathcal{J}$ and $j$.
    From Lemma~\ref{lemma:M construction} and (\ref{eq:add to Mset}) directly follows that ${\tilde{\mathcal M}_i(j,\mathcal J) \cap \mathbb Z_j = \emptyset}$, which implies with (\ref{eq:Mset}) that 
    \begin{equation*}
        {\mathcal{M}_i(\mathcal{J}) \cap \mathbb{Z}_j} = {\mathcal{M}_i(\mathcal{J} \cup \{j\}) \cap \mathbb{Z}_j}.
    \end{equation*}
    In other words, all indices of the additional balancing rows $\tilde{\mathcal{M}}_i(j,\mathcal{J})$ for $\mathcal {M}_i(\mathcal{J} \cup \{j\})$ are greater than $j$.
    Correspondingly, by using (\ref{eq:Ui set}) we find  for the message vectors 
    \begin{equation*}
        \boldsymbol{u}_{\mathbb{Z}_j}(i,\mathcal{J}) = \boldsymbol{u}_{\mathbb{Z}_j}(i,\mathcal{J} \cup \{j\}).
    \end{equation*}
    Further, $u_j(i,\mathcal{J}) < u_j(i,\mathcal{J} \cup \{j\})$, with which follows that $\mathcal{L}_2=\mathcal{K}_i$.
    From (\ref{eq:num wmin Qiwmin}) we know that $\left|\mathcal Q_{i,\wmin}\right| = 2^{\left|\mathcal K_i\right|}$.
    Consequently, following Lemma~\ref{proposition:dichotomous binary tree}, $\mathcal{Q}_{i,\wmin}$ can be represented as a \ac{PDBT}.
\end{proof}

Building on Proposition~\ref{proposition: Qiwmin tree}, we can now determine how to efficiently compute the message tree corresponding to $\mathcal{Q}_{i,\wmin}$.
\begin{proposition}\label{proposition:update message}
    Let $\mathcal{J} \cup \{j\} \subseteq \mathcal{K}_i$ and $j > \max(\mathcal{J})$. Then a message ${\boldsymbol{u}(i,\mathcal{J})} \in \mathcal{Q}_{i,\wmin}$ can be computed as 
    \begin{equation}
        \begin{split}\label{eq:Ui set successive}
        \operatorname{supp}\left(\boldsymbol{u}(i,\mathcal{J} \cup \{j\})\right) &= \mathcal{U}_i(\mathcal{J} \cup \{j\})\\
        &= \left(\mathcal{U}_i(\mathcal{J}) \setxor \tilde{\mathcal{M}}_i(j,\mathcal{J}) \right) \cup \{j\},
        \end{split}
    \end{equation}       
    where the additional balancing rows are found as
    \begin{equation*}\label{eq:add to Mset successive}
        \tilde{\mathcal{M}}_i(j,\mathcal {J}) = \underset{{\substack{k \in \mathcal{U}_i(\mathcal{J}) \cap \mathbb{Z}_{i+1:j},\\\bar{i} \land j \land k = 0}}}{\bigsetxor} \left\{\mu_i(j,k)\right\}.  
    \end{equation*}
\end{proposition}
\begin{proof}
    Let ${x<y<z}$, then
    with $l_x^*(y) = \max\left(\mathcal{S}_x^\mathrm{c} \cap \mathcal{S}_y\right)$
    follows 
    $h_{(l)}=1$ for all $l \in \mathcal{S}_x \cap \mathbb{Z}_{l_x^*(y):n}$ and thus $l_x^*(z) \geq l_x^*(y)$.
    Consequently, for $j > \max\left(\mathcal J\right)$ follows $l_i^*(j) \geq l_i^*(\max(\mathcal J))$.
    Let $i < \min\{j,k\}$, then we find $l_i^*\left(\mu_i(j,k)\right) = l_i^*(\max\{j,k\})$.
    With (\ref{eq:Mset}) and (\ref{eq:add to Mset}) follows $l_i^*(j) \geq l_i^*(\max(\mathcal M_i(\mathcal J)))$. 
    Let ${\bar{i} \land j \land k = 0}$, then $l_i^*(j) \geq l_i^*(k)$ implies $j>k>i$.
    Thus, if $j > \max\left(\mathcal J\right)$, only rows $k \in \mathcal{U}_i(\mathcal{J}) \cap \mathbb{Z}_{i+1:j}$ have to be considered for the computation of $\tilde{\mathcal{M}}_i(j,\mathcal {J})$ with (\ref{eq:add to Mset}).
    Subsequently, (\ref{eq:Ui set successive}) follows from (\ref{eq:Ui set}) and (\ref{eq:Mset}).
\end{proof}

Following Proposition~\ref{proposition: Qiwmin tree} and Proposition~\ref{proposition:update message}, the messages $\boldsymbol{u} \in \mathcal{Q}_{i,\wmin}$ can be found without redundant operations from the tree.
Although both sets $\Qi$ and $\mathcal{Q}_{i,\wmin}$ can be represented as \acp{PDBT}, the message calculation differs.
For $\Qi$, the value of an einzelchild node is computed at the corresponding einzelchild level.
Thus, we only obtain a message $\boldsymbol{u} \in \Qi$ after we have followed a complete path through the tree.
In contrast, for $\mathcal{Q}_{i,\wmin}$, the values of the einzelchild nodes are determined at the sibling levels.
Hence, following Proposition~\ref{proposition:update message}, we directly find another message $\boldsymbol{u} \in \mathcal{Q}_{i,\wmin}$ if we update the set of core rows.

Note that the message tree of $\Qi$ could be computed similarly to the tree corresponding to $\mathcal{Q}_{i,\wmin}$.
But as it will turn out, the proposed method is empirically more efficient for computing the intersection tree of $\Qi$ and $\mathcal{Q}_{i,\wmin}$.

\subsection{Tree Intersection}
With the knowledge of the structure of the trees corresponding to the messages sets $\Qi$ and $\mathcal{Q}_{i,\wmin}$, we can finally deal with the computation of the intersection tree and subsequently the calculation of (\ref{eq:intersection cardinality}).
First, we determine the type of the intersection tree, using the fact that both message trees are \acp{PDBT}.

\begin{proposition}\label{proposition:tree intersection}
    The intersection tree $\mathcal{Z}$ of two \acp{DBT} $\mathcal{X}$, $\mathcal{Y}$ of height $N$ 
    with einzelchild levels $\mathcal{L}_{1,\mathcal{X}}$, $\mathcal{L}_{1,\mathcal{Y}}$,
    and sibling levels $\mathcal{L}_{2,\mathcal{X}}$, $\mathcal{L}_{2,\mathcal{Y}}$
    is again a \ac{DBT} with einzelchild levels
    \begin{equation*}
        \mathcal{L}_{1,\mathcal{Z}} = \left(\mathcal{L}_{1,\mathcal{X}} \cap \mathcal{L}_{1,\mathcal{Y}}\right) \cup \left(\mathcal{L}_{1,\mathcal{X}} \cap \mathcal{L}_{2,\mathcal{Y}}\right) \cup \left(\mathcal{L}_{2,\mathcal{X}} \cap \mathcal{L}_{1,\mathcal{Y}}\right)
    \end{equation*}
and sibling levels $\mathcal{L}_{2,\mathcal{Z}} = \mathcal{L}_{2,\mathcal{X}} \cap \mathcal{L}_{2,\mathcal{Y}}$.
\end{proposition}
\begin{proof}
    If the root node of $\mathcal{X}$ and $\mathcal{Y}$ is equal, the intersection tree $\mathcal{Z}$ has the same root node.
    Otherwise, $\mathcal{Z}$ is empty.
    From the root node onward, $\mathcal{Z}$ is found successively level by level.
    A node in $\mathcal{Z}$ has a left and/or right child node if and only if both $\mathcal{X}$ and $\mathcal{Y}$ contain the corresponding node. 
    As $\mathcal{X}$ and $\mathcal{Y}$ are both dichotomous, there are four cases of levels:
    \begin{enumerate}
        \item $l \in \mathcal{L}_{1,\mathcal{X}} \cap \mathcal{L}_{1,\mathcal{Y}}$:
        A node in $\mathcal{Z}$ at level $l-1$ has the same child node as the corresponding nodes in $\mathcal{X}$ and in $\mathcal{Y}$ if their child nodes coincide. Otherwise, it is a leaf node.    
        \item $l \in \mathcal{L}_{1,\mathcal{X}} \cap \mathcal{L}_{2,\mathcal{Y}}$:
        A node in $\mathcal{Z}$ at level $l-1$ has the same child node as the corresponding node in $\mathcal{X}$.
        \item $l \in \mathcal{L}_{2,\mathcal{X}} \cap \mathcal{L}_{1,\mathcal{Y}}$:
        A node in $\mathcal{Z}$ at level $l-1$ has the same child node as the corresponding node in $\mathcal{Y}$.
        \item $l \in \mathcal{L}_{2,\mathcal{X}} \cap \mathcal{L}_{2,\mathcal{Y}}$:
        All nodes in $\mathcal{Z}$ are sibling nodes.
    \end{enumerate}
    Subsequently, cases 1) to 3) correspond to the einzelchild levels $\mathcal{L}_{1,\mathcal{Z}}$ of $\mathcal{Z}$ and case 4) to the sibling levels $\mathcal{L}_{2,\mathcal{Z}}$.
    Therefore, $\mathcal{Z}$ is dichotomous.
\end{proof}

From Proposition~\ref{proposition:tree intersection} follows that the intersection tree of the \acp{PDBT} corresponding to $\Qi$ and $\mathcal{Q}_{i,\wmin}$ is a \ac{DBT} with sibling levels $\mathcal{L}_2 = \mathcal{I} \cap \mathcal{K}_i$. 
Subsequently, to find a message $\boldsymbol{u} \in \Qiwmin$, 
the bits at levels $\mathcal{I} \cap \mathcal{K}_i$ can be chosen freely.
The bits at levels $\mathcal{F} \cap \mathcal{K}_i$ are determined by the pre-transformation.
For ``1'' bits at these levels $\mathcal{K}_i$, $\boldsymbol{u}$ is updated as in Proposition~\ref{proposition:update message}, whereby the nodes of levels $\mathcal{K}_i^\mathrm{c}$ are predetermined. 
The corresponding bits at levels $\mathcal{I} \cap \mathcal{K}_i^\mathrm{c}$ can be chosen accordingly.
For bits $\mathcal{F} \cap \mathcal{K}_i^\mathrm{c}$, both trees of  $\Qi$ and $\mathcal{Q}_{i,\wmin}$ have einzelchild levels.
Consequently, the current message $\boldsymbol{u}$ and the pre-transformation have to match.
If not, then we found a leaf node of the intersection tree and the message path does not form a $\wmin$-weight codeword of the code.

Conversely, for level 
$f_i^*(\mathcal{I}) = \max\left(\{i\} \cup \left(\mathcal{K}_i^\mathrm{c} \setminus \mathcal{I} \right)\right)$,
the subsequent sub-tree is a \ac{PDBT} with sibling levels $\mathcal{L}_2={\mathcal{I} \cap \mathcal{K}_i^\circ(\mathcal{I})}$, $\smash{{\mathcal{K}_i^\circ(\mathcal{I}) = \mathcal{K}_i \setminus \mathbb{Z}_{f_i^*(\mathcal{I})}}}$, 
and we can write the enumeration as
\begin{equation}\label{eq:Awmin subtree simplification}
    \left|\Qiwmin\right| = \left|\mathcal{Q}_{i,\wmin}(\mathcal{I} \setminus \mathcal{K}_i^\circ(\mathcal{I}),\boldsymbol{T})\right| \cdot 2^{\left|\mathcal{I} \cap \mathcal{K}_i^\circ(\mathcal{I})\right|}.
\end{equation}
Following this, Algorithm~\ref{alg:error_coeff} recursively traverses the intersection tree in \emph{reverse pre-order} and counts the number of $\wmin$-weight codewords\footnote{Implementation is available: \url{https://github.com/andreaszunker/PTPC}}.

\begin{algorithm}[t]
    \small%
    \SetAlgoLined\LinesNumbered
    \SetKwInOut{Input}{Input}\SetKwInOut{Output}{Output}
    \Input{Rate-profile $\mathcal I$, pre-transformation matrix $\boldsymbol T$}
    \Output{Number of $\wmin$-weight codewords $A_{\wmin}$}
    $\boldsymbol{T}_\mathcal{I} \gets \operatorname{RREF}(\boldsymbol{T}_\mathcal{I}),\;\boldsymbol{T}_\mathcal{F} \gets \boldsymbol{0}$\;
    \Return $\sum_{i \in \mathcal I_{\wmin}} \operatorname{enumerate\_subtree}(i,i,\boldsymbol{0},\mathcal {I},\boldsymbol{T}) \cdot 2^{|\mathcal{I} \cap \mathcal{K}_i^\circ(\mathcal{I})|}$\;
    \algrule[.5pt]
    \Fn{$\operatorname{enumerate\_subtree}(i,j,\boldsymbol u,\mathcal {I},\boldsymbol{T})$}{
    $\Awmin \gets 0$\;
    	$\boldsymbol{u} \gets \operatorname{update\_message}(i,j,\operatorname{copy}(\boldsymbol{u}))$\;
        \For{$k \gets j+1$ \KwTo $f_i^*(\mathcal{I})$}{
            \uIf{$k \in \mathcal{I} \cap \mathcal{K}_i$}{
                $\Awmin \gets \Awmin + \operatorname{enumerate\_subtree}(i,k,\boldsymbol u,\mathcal {I},\boldsymbol{T})$\;
            }
            \ElseIf{$k \in \mathcal{F}$  \rm{\textbf{and}} $u_k \neq \bigoplus_{h=i}^{k-1} u_h \cdot t_{h,k}$}{
                \uIf{$k \in \mathcal{K}_i$}{
                    $\boldsymbol{u} \gets \operatorname{update\_message}(i,k,\boldsymbol{u})$\;
                }
                \lElse{\Return $\Awmin$}
            }
        }
        \Return $\Awmin$ + 1\;
    }
    \algrule[.5pt]
    \Fn{$\operatorname{update\_message}(i,j,\boldsymbol{u})$}{
        \For{$k \gets i+1$ \KwTo $j-1$\label{algline:i+1 to j-1}}{
            \If{$u_k = 1$ \rm{\textbf{and}} $\bar i \land j \land k = 0$}{
                $u_{\mu_i(j,k)} \gets u_{\mu_i(j,k)} \oplus 1$\;
            }
        }
        $u_j \gets 1$\;
        \Return $\boldsymbol{u}$\;
    }
	\caption{\footnotesize
    Enumeration of the $\wmin$-weight codewords
    }
    \label{alg:error_coeff}
\end{algorithm}
To illustrate the structure of the message trees corresponding to $\Qi$ and $\mathcal{Q}_{i,\wmin}$ as well as their intersection, we give the following example.
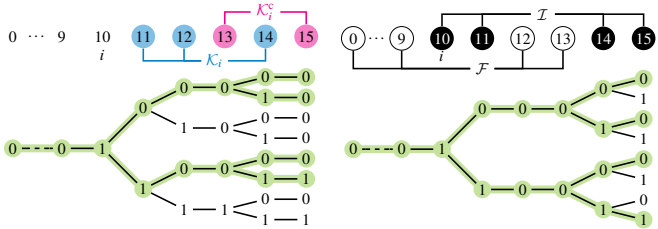
\begin{figure}[h]
	\centering
    \begin{subfigure}[b]{0.48\textwidth}
                \resizebox{\columnwidth}{!}%
                {\begin{tikzpicture}[yscale=0.45,xscale=0.9]
\tikzstyle{basic}=[circle, draw=none,fill=white,inner sep=1pt]
\tikzstyle{hlnode}=[circle,draw=none,fill=apfelgruen!50,inner sep=1pt]
\tikzstyle{head}=[circle,draw=none,fill=none,inner sep=1pt]
\tikzstyle{line} = [line width = 1pt]
\tikzstyle{highlight} = [line width = 5pt, apfelgruen!50]

\node[head] at (0.8,0) {$0$};
\node[head] at (1.4,0) {$\dots$};
\node[head] at (2,0) {$9$};
\node[head] at (3,0) {$10$};
\node[head, fill=mittelblau!50] (ki1) at (4,0) {$11$};
\node[head, fill=mittelblau!50] (ki2) at (5,0) {$12$};
\node[head, fill=pink!50] (kic1) at (6,0) {$13$};
\node[head, fill=mittelblau!50] (ki3) at (7,0) {$14$};
\node[head, fill=pink!50] (kic2) at (8,0) {$15$};
\node[head] at (3,-.9) {$i$};

\node[head, mittelblau] (ki) at (5.75,-1.2) {$\mathcal{K}_i$};
\draw[line, mittelblau] (ki1) |- (ki);
\draw[line, mittelblau] (ki2) |- (ki);
\draw[line, mittelblau] (ki3) |- (ki);

\node[head, pink] (kic) at (7,1.2) {$\mathcal{K}^\mathrm{c}_i$};
\draw[line, pink] (kic1) |- (kic);
\draw[line, pink] (kic2) |- (kic);

\coordinate (r0) at (0.8,-5.5);
\coordinate (r1a) at (1.1,-5.5);
\coordinate (r1b) at (1.7,-5.5);
\coordinate (r2) at (2,-5.5);
\coordinate (r3) at (3,-5.5);

\coordinate (r4a) at (4,-3.5);
\coordinate (r4b) at (4,-7.5);

\coordinate (r5a) at (5,-2.5);
\coordinate (r5b) at (5,-4.5);
\coordinate (r5c) at (5,-6.5);
\coordinate (r5d) at (5,-8.5);

\coordinate (r6a) at (6,-2.5);
\coordinate (r6b) at (6,-4.5);
\coordinate (r6c) at (6,-6.5);
\coordinate (r6d) at (6,-8.5);

\coordinate (r7a) at (7,-2);
\coordinate (r7b) at (7,-3);
\coordinate (r7c) at (7,-4);
\coordinate (r7d) at (7,-5);

\coordinate (r7e) at (7,-6);
\coordinate (r7f) at (7,-7);
\coordinate (r7g) at (7,-8);
\coordinate (r7h) at (7,-9);

\coordinate (r8a) at (8,-2);
\coordinate (r8b) at (8,-3);
\coordinate (r8c) at (8,-4);
\coordinate (r8d) at (8,-5);

\coordinate (r8e) at (8,-6);
\coordinate (r8f) at (8,-7);
\coordinate (r8g) at (8,-8);
\coordinate (r8h) at (8,-9);

\draw[highlight] (r0) -- (r2) -- (r3);
\draw[highlight] (r3) -- (r4a) -- (r5a) -- (r6a) -- (r7a) -- (r8a);
\draw[highlight] (r3) -- (r4b) -- (r5c) -- (r6c) -- (r7e) -- (r8e);
\draw[highlight] (r6a) -- (r7b) -- (r8b);
\draw[highlight] (r6c) -- (r7f) -- (r8f);

\draw[line] (r6b) -- (r7d) -- (r8d);
\draw[line] (r6d) -- (r7h) -- (r8h);
\draw[line] (r4a) -- (r5b) -- (r6b) -- (r7c) -- (r8c);
\draw[line] (r4b) -- (r5d) -- (r6d) -- (r7g) -- (r8g);
\draw[line] (r0) -- (r1a);
\draw[line] (r1b) -- (r2) -- (r3);
\draw[line] (r3) -- (r4a) -- (r5a) -- (r6a) -- (r7a) -- (r8a);
\draw[line] (r3) -- (r4b) -- (r5c) -- (r6c) -- (r7e) -- (r8e);
\draw[line] (r6a) -- (r7b) -- (r8b);
\draw[line] (r6c) -- (r7f) -- (r8f);
\draw[line, dashed] (r1a) -- (r1b);

\node[hlnode] at (r0)  {$0$};
\node[hlnode] at (r2)  {$0$};
\node[hlnode] at (r3)  {$1$};
\node[hlnode] at (r4a) {$0$};
\node[hlnode] at (r4b) {$1$};
\node[hlnode] at (r5a) {$0$};
\node[basic] at (r5b) {$1$};
\node[hlnode] at (r5c) {$0$};
\node[basic] at (r5d) {$1$};

\node[hlnode] at (r6a) {$0$};
\node[basic] at (r6b) {$0$};
\node[hlnode] at (r6c) {$0$};
\node[basic] at (r6d) {$1$};
\node[hlnode] at (r7a) {$0$};
\node[hlnode] at (r7b) {$1$};
\node[basic] at (r7c) {$0$};
\node[basic] at (r7d) {$1$};
\node[hlnode] at (r7e) {$0$};
\node[hlnode] at (r7f) {$1$};
\node[basic] at (r7g) {$0$};
\node[basic] at (r7h) {$1$};
\node[hlnode] at (r8a) {$0$};
\node[hlnode] at (r8b) {$0$};
\node[basic] at (r8c) {$0$};
\node[basic] at (r8d) {$0$};
\node[hlnode] at (r8e) {$0$};
\node[hlnode] at (r8f) {$1$};
\node[basic] at (r8g) {$0$};
\node[basic] at (r8h) {$1$};

\end{tikzpicture}}
                \caption{ \ac{PDBT} of $\mathcal{Q}_{i=10,\wmin=4}$, branching before each bit $u_h$, $h \in \mathcal{L}_2 = \mathcal{K}_{i=4}$.\newline
                }
                \label{subfig:tree_universal}
        \end{subfigure}
        \hfill
        \begin{subfigure}[b]{0.48\textwidth}
                \resizebox{\columnwidth}{!}%
                {\begin{tikzpicture}[yscale=0.45,xscale=0.9]
\tikzstyle{basic}=[circle, draw=none,fill=white,inner sep=1pt]
\tikzstyle{hlnode}=[circle,draw=none,fill=apfelgruen!50,inner sep=1pt]
\tikzstyle{head}=[circle,draw=none,fill=none,inner sep=1pt, minimum size=15pt]
\tikzstyle{line} = [line width = 1pt]
\tikzstyle{highlight} = [line width = 5pt, apfelgruen!50]

\node[head, draw=black] (f1) at (0.8,0) {$0$};
\node[head] at (1.4,0) {$\dots$};
\node[head, draw=black] (f2) at (2,0) {$9$};
\node[head, fill=black, text=white] (i1) at (3,0) {$10$};
\node[head, fill=black, text=white] (i2) at (4,0) {$11$};
\node[head, draw=black] (f3) at (5,0) {$12$};
\node[head, draw=black] (f4) at (6,0) {$13$};
\node[head, fill=black, text=white] (i3) at (7,0) {$14$};
\node[head, fill=black, text=white] (i4) at (8,0) {$15$};
\node[head] at (3,-.9) {$i$};

\node[head] (i) at (5.5,1.2) {$\mathcal{I}$};
\draw[line] (i1) |- (i);
\draw[line] (i2) |- (i);
\draw[line] (i3) |- (i);
\draw[line] (i4) |- (i);

\node[head] (f) at (4,-1.5) {$\mathcal{F}$};
\draw[line] (f1) |- (f);
\draw[line] (f2) |- (f);
\draw[line] (f3) |- (f);
\draw[line] (f4) |- (f);

\coordinate (r0) at (0.8,-5.5);
\coordinate (r1a) at (1.1,-5.5);
\coordinate (r1b) at (1.7,-5.5);
\coordinate (r2) at (2,-5.5);
\coordinate (r3) at (3,-5.5);

\coordinate (r4a) at (4,-3.5);
\coordinate (r4b) at (4,-7.5);

\coordinate (r5a) at (5,-3.5);
\coordinate (r5b) at (5,-7.5);

\coordinate (r6a) at (6,-3.5);
\coordinate (r6b) at (6,-7.5);

\coordinate (r7a) at (7,-2.5);
\coordinate (r7b) at (7,-4.5);
\coordinate (r7c) at (7,-6.5);
\coordinate (r7d) at (7,-8.5);

\coordinate (r8a) at (8,-2);
\coordinate (r8b) at (8,-3);
\coordinate (r8c) at (8,-4);
\coordinate (r8d) at (8,-5);

\coordinate (r8e) at (8,-6);
\coordinate (r8f) at (8,-7);
\coordinate (r8g) at (8,-8);
\coordinate (r8h) at (8,-9);

\draw[highlight] (r0) -- (r2) -- (r3);
\draw[highlight] (r3) -- (r4a) -- (r5a) -- (r6a) -- (r7a) -- (r8a);
\draw[highlight] (r3) -- (r4b) -- (r5b) -- (r6b) -- (r7c) -- (r8e);
\draw[highlight] (r6a) -- (r7b) -- (r8c);
\draw[highlight] (r6b) -- (r7d) -- (r8h);

\draw[line] (r0) -- (r1a);
\draw[line] (r1b) -- (r2) -- (r3);
\draw[line] (r3) -- (r4a) -- (r5a) -- (r6a) -- (r7a) -- (r8a);
\draw[line] (r3) -- (r4b) -- (r5b) -- (r6b) -- (r7c) -- (r8e);
\draw[line] (r6a) -- (r7b) -- (r8c);
\draw[line] (r6b) -- (r7d) -- (r8h);
\draw[line] (r7a) -- (r8b);
\draw[line] (r7b) -- (r8d);
\draw[line] (r7c) -- (r8f);
\draw[line] (r7d) -- (r8g);
\draw[line, dashed] (r1a) -- (r1b);

\node[hlnode] at (r0)  {$0$};
\node[hlnode] at (r2)  {$0$};
\node[hlnode] at (r3)  {$1$};
\node[hlnode] at (r4a) {$0$};
\node[hlnode] at (r4b) {$1$};
\node[hlnode] at (r5a) {$0$};
\node[hlnode] at (r5b) {$0$};

\node[hlnode] at (r6a) {$0$};
\node[hlnode] at (r6b) {$0$};
\node[hlnode] at (r7a) {$0$};
\node[hlnode] at (r7b) {$1$};
\node[hlnode] at (r7c) {$0$};
\node[hlnode] at (r7d) {$1$};
\node[hlnode] at (r8a) {$0$};
\node[basic] at (r8b) {$1$};
\node[hlnode] at (r8c) {$0$};
\node[basic] at (r8d) {$1$};
\node[hlnode] at (r8e) {$0$};
\node[basic] at (r8f) {$1$};
\node[basic] at (r8g) {$0$};
\node[hlnode] at (r8h) {$1$};

\end{tikzpicture}}
                \caption{ \ac{PDBT} of $\mathcal{Q}_{i=10}(\mathcal{I})$, with non-decreasing rate-profile ${\mathcal{I}}$, %
                branching before each bit $u_h$, $h \in \mathcal{L}_2 = \mathcal{I} \setminus \mathbb{Z}_{i+1}$.}
                \label{subfig:tree_rate_profile}
        \end{subfigure}
        \vspace{-.2cm}
    \caption{
        Message trees of a universal polar coset (a) and a plain polar coset (b) for $N=16$ and $i=10$. The intersection of the trees (i.e., common messages $\boldsymbol{u}$) is highlighted. In this example, $\left|\mathcal{Q}_{i=10,\wmin=4}(\mathcal{I})\right|=4$.
    }
    \label{fig:weightenumtree}
\end{figure}

\begin{example}
    \normalfont Consider the case $N=16$ and $\wmin=4$. We regard the coset $\mathcal P_{i=10}$. Using (\ref{eq:ki}), we find $\mathcal{K}_{i=10}=\{11,12,14\}$ and $\mathcal{K}^\mathrm{c}_{i=10}=\{13,15\}$. Hence, there are $2^{|\mathcal{K}_{i=10}|}=8$ weight-4 codewords in this coset. The corresponding \ac{PDBT} of the messages generating these $\wmin$-weight codewords is shown in Fig.~\ref{subfig:tree_universal}.
    Now consider the rate-profile ${\mathcal{I}=\{10,11,14,15\}}$. The coset $\mathcal{Q}_{i=10}(\mathcal{I})$ forms the \ac{PDBT} depicted in Fig.~\ref{subfig:tree_rate_profile}. The intersection of both trees is $\mathcal{Q}_{i=10,\wmin=4}(\mathcal{I})$ and highlighted; it contains 4 weight-4 codewords. In this case, the leaf nodes of the intersection tree are at the level $N-1$.
\end{example}

\section{Theoretical Findings}\label{sec:findings}
If we find $\mathcal{F} \cap \mathcal{K}_i^\mathrm{c} = \emptyset$ for a coset $\Ci$ where $i \in \mathcal{I}_{\wmin}$, it follows that $f_i^*(\mathcal{I}) = i$ and thus $\mathcal{K}_i^\circ(\mathcal{I})=\mathcal{K}_i$.
Consequently, from (\ref{eq:Awmin subtree simplification}) we can derive that $\left|\Qiwmin\right| = 2^{\left|\mathcal{I} \cap \mathcal{K}_i\right|}$, meaning the number of $\wmin$-weight codewords of coset $\Ci$ is independent of $\boldsymbol{T}$, and therefore $\Awmin(\Ci)=\Awmin(\mathcal{P}_i(\mathcal{I}))$.
Subsequently, as introduced in \cite{Rowshan2023Impact}, we distinguish two types of cosets with indices $i \in \mathcal{I}_{\wmin}$, based on whether $\Awmin(\Ci)$ can be reduced by choosing a suitable pre-transformation $\boldsymbol{T}$:
\begin{equation*}
    \Ci \mathrel{\widehat{=}} 
    \begin{cases}
        \emph{``pre-transformable"} & \text{ 
        if } \mathcal{F} \cap \mathcal{K}_i^\mathrm{c} \neq \emptyset\\
        \emph{``non pre-transformable"} & 
        \text{ otherwise.}
    \end{cases}
\end{equation*}
Correspondingly, the set of indices of pre-transformable cosets is $\mathcal{I}^*_{\wmin} = \left\{i \in \mathcal I_{\wmin} \,\middle|\, \mathcal{F} \cap \mathcal{K}_i^\mathrm{c} \neq \emptyset\right\}$ and of non pre-transformable cosets is $\mathcal{I}^\circ_{\wmin} = \mathcal{I}_{\wmin} \setminus \mathcal{I}^*_{\wmin}$.
Using this classification and the previous findings, we can characterize a \ac{PTPC} $\C$ based on the chosen rate-profile $\mathcal{I}$.

Firstly, we consider the minimum distance $\dmin$ of a  \ac{PTPC} $\C$ with decreasing and with $\RM(r,n)$ rate-profile $\mathcal{I}$.
\begin{theorem}\label{thm:dmin=wmin}
    A pre-transformed polar code $\C$ with a decreasing rate-profile $\mathcal{I}$ has the same minimum distance as the corresponding plain polar code $\mathcal P(\mathcal{I})$.
    \label{thm:dmin}
\end{theorem}
\begin{proof}
    As $\mathcal{I}$ is decreasing, from the \textit{left swap} rule follows that $\mathcal{P}(\mathcal {I})$ includes the most reliable $\wmin$-weight row $\boldsymbol g_{i}$ with index ${i = \max\left(\mathcal I_{\wmin}\right) = N - \nicefrac{N}{\wmin}}$
    as $i \succcurlyeq j \in \mathcal I_{\wmin}$.
    Subsequently, we find $\mathbb{Z}_{i:N}=\left\{j \in \mathbb Z_N \,\middle|\, i \preccurlyeq j \right\} \subseteq \mathcal{I}$ by
    following the \textit{binary domination} rule.
    Thus, there are no frozen rows $\boldsymbol{g}_f$ with $f > i$ to add to $\boldsymbol{g}_i$, implying that $\boldsymbol g_{i} \in \mathcal P(\mathcal I,\boldsymbol T)$.
    With the lower bound~(\ref{eq:dmin lower bound}), it follows that $\dmin(\C) = \dmin(\mathcal{P}(\mathcal{I})) = \wmin$.
\end{proof}
\vspace{0.01cm}

\begin{corollary}\label{cor:dmin rm}
    A pre-transformed polar code $\C$ with an $\RM(r,n)$ rate-profile $\mathcal{I}$ has minimum distance $\dmin = 2^{n-r}$. 
\end{corollary}
\begin{proof}
    From the $\RM(r,n)$ rate-profile $\mathcal{I}$ follows that $\mathcal{I}$ is decreasing and that $\mathcal{P}(\mathcal{I}) = \RM(r,n)$.
    Thus, with Theorem~\ref{thm:dmin=wmin} we find $\dmin\left(\C\right) = \dmin\left(\mathcal{P}(\mathcal{I})\right) = 2^{n-r}$.
\end{proof}
\vspace{0.01cm}
Secondly, we examine the number of minimum weight codewords $\Awmin$ of a \ac{PTPC} $\C$ with arbitrary and with $\RM(r,n)$ rate-profile $\mathcal{I}$.
\begin{lemma}\label{lemma:lower bound decreasing monomial} (\hspace{1sp}\cite[Lemma 3]{Rowshan2023Impact})
    The number of $\wmin$-weight codewords of a pre-transformed polar code $\C$ with arbitrary rate-profile $\mathcal{I}$ is lower bounded as
    \begin{equation*}
        \Awmin(\C) \geq A_{\wmin}^{\mathrm{LB}}(\mathcal I) = \sum\limits_{i \in \mathcal I_{\wmin}^\circ} 2^{\left|\mathcal{I} \cap \mathcal{K}_i\right|}.
    \end{equation*}
\end{lemma}
\begin{proof}
    As we saw above, for a cosets with index $i \in \mathcal{I}^\circ_{\wmin}$ holds ${\mathcal{F} \cap \mathcal{K}_i^\mathrm{c} = \emptyset}$ and thus $\left|\Qiwmin\right| = 2^{\left|\mathcal{I} \cap \mathcal{K}_i\right|}$. By only considering these non pre-transformable cosets in~(\ref{eq:Awmin}), we obtain a lower bound on $\Awmin(\C)$.
\end{proof}

For an $\RM(r,n)$ rate-profile $\mathcal{I}$ we can find a closed form solution of the lower bound on the number of $\wmin$-weight codewords of a \ac{PTPC} $\C$ given in Lemma~\ref{lemma:lower bound decreasing monomial}.
First, we consider which cosets are non pre-transformable. 

\begin{lemma}\label{lem:rm non pretransformable}
A coset $\Ci$ with $\RM(r,n)$ rate-profile $\mathcal{I}$ of order $r\leq n-2$ is non pre-transformable if the support of its index $i \in \mathcal{I}_{\dmin}$ is $\mathcal{S}_{i} = \{x,y\} \cup \mathbb Z_{r+2:n}$, where $0 \leq x < y < r-2$.
\end{lemma}
\begin{proof}
Following the $\RM(r,n)$ rate-profile $\mathcal{I}$, for $f \in \mathcal{F}$ holds $\w(f) < n-r$.
To fulfill $f > i$, the support of $f$ is found as $\mathcal{S}_f = \{z\} \cup \mathbb{Z}_{r+2:n}$ with $z > y$. 
Subsequently, $\mathcal{S}_f \setminus \mathcal{S}_i = \{z\}$ and thus $\w\left(\bar{i} \land f\right) = 1$, implying that $\mathcal{F} \cap \mathcal{K}_i^\mathrm{c} = \emptyset$.
By definition, $\Ci$ is non pre-transformable.
\end{proof}

Knowing which cosets are non pre-transformable, we can now find a closed form lower bound on $\Admin(\C)$ for an $\RM(r,n)$ rate-profile $\mathcal{I}$.

\begin{theorem}\label{thm:rmbound}
The number of codewords with weight $\dmin=2^{n-r}$ of a pre-transformed polar code $\mathcal P(\mathcal I,\,\boldsymbol T)$ with $\RM(r,n)$ rate-profile $\mathcal{I}$ of order $r\leq n-2$ is lower bounded as
\begin{equation}\label{eq:lowerbound_rm}
    A_{\dmin}^{\mathrm{LB}}(r) =
    \frac{8\cdot2^{3r} - 6\cdot2^{2r} + 2^r}{3}.
\end{equation}
\end{theorem}
\begin{proof}
    From Corollary~\ref{cor:dmin rm} we know that $\dmin(\C) = 2^{n-r}$ for an $\RM(r,n)$ rate-profile $\mathcal{I}$.
    Following Lemma~\ref{lem:rm non pretransformable}, for order $r \leq n-2$,
    cosets having an index support $\mathcal{S}_{i} = \{x,y\} \cup \mathbb Z_{r+2:n}$, where $0 \leq x < y < r+2$, are non pre-transformable.
    As we have determined above, $\Admin(\Ci) = 2^{\left|\mathcal{I} \cap \mathcal{K}_i\right|}$ holds for these cosets $i \in \mathcal{I}_{\smash{\dmin}}' \subseteq \mathcal{I}_{\smash{\dmin}}^\circ$.
    As shown in \cite{Rowshan2023Impact}, for a decreasing (and thus also for \ac{RM}) rate-profile $\mathcal{I}$ follows
    \begin{equation*}
        \left|\mathcal{I} \cap \mathcal{K}_i\right|=\left|\mathcal S_i^\mathrm{c}\right|+\sum_{l \in \mathcal S_i} \left|\mathcal S_i^\mathrm{c} \setminus \mathbb Z_l\right|,
    \end{equation*}
    simply put, the number of ``0'' bits of coset index $i$ plus the number of more significant ``0'' bits for each ``1'' bit.
    For ${i \in  \mathcal {I}_{\smash{\dmin}}'}$ follows $\left|\mathcal S_i^\mathrm{c}\right|=n-\left|\mathcal S_i\right|=r$.
    Since ${\mathcal S_i^\mathrm{c} \subseteq \mathbb Z_{r+2}}$ and ${l \geq r+2}$ imply that $\mathcal S_i^\mathrm{c} \setminus \mathbb Z_l=\emptyset$, we find for $i \in \mathcal{I}'_{\smash{\dmin}}$ that
    \begin{equation*}
        \sum_{l \in \mathcal{S}_i} \left|\mathcal{S}_i^\mathrm{c} \setminus \mathbb{Z}_l\right| = \left|\mathcal{S}_i^\mathrm{c} \setminus \mathbb{Z}_x\right|+\left|\mathcal{S}_i^\mathrm{c} \setminus \mathbb{Z}_y\right|.
    \end{equation*}
    Further, the number of ``0'' bit positions larger than $x$ and $y$ is found as ${\left|\mathcal{S}_i^\mathrm{c} \setminus \mathbb{Z}_x\right| = r-x}$ and ${\left|\mathcal{S}_i^\mathrm{c} \setminus \mathbb{Z}_y\right| = r-y+1}$, respectively.
    Thus, we obtain $\left|\mathcal{I} \cap \mathcal{K}_i\right| = 3r+1-x-y$, and find with Lemma~\ref{lemma:lower bound decreasing monomial} the lower bound 
    \begin{equation*}
        \Admin(\C) \geq 
        \sum\limits_{i \in \mathcal{I}_{\dmin}^\circ} 2^{\left|\mathcal{I} \cap \mathcal{K}_i\right|} \geq
        \sum\limits_{i \in \mathcal{I}_{\dmin}'} 2^{\left|\mathcal{I} \cap \mathcal{K}_i\right|} = A_{\dmin}^{\mathrm{LB}}(r),
    \end{equation*}
    where
    \begin{equation*}
        A_{\dmin}^{\mathrm{LB}}(r) = \sum_{y=1}^{r+1} \sum_{x=0}^{y-1} 2^{3r+1-x-y} = 2^{3r+1} \cdot \sum_{y=1}^{r+1} \frac{1}{2^y} \cdot \sum_{x=0}^{y-1} \frac{1}{2^x}.
    \end{equation*}
    By employing the geometric series, we obtain the closed-form solution
    \begin{equation*}
    \begin{split}
        A_{\dmin}^{\mathrm{LB}}(r) 
        &= 2^{3r+2} \cdot \sum_{y=1}^{r+1} \frac{1}{2^y} \cdot \left(1-\frac{1}{2^y}\right)
        =
        2^{3r+2} \cdot \left(\sum_{y=0}^{r+1} \frac{1}{2^y} -  \sum_{y=0}^{r+1} \frac{1}{2^{2y}}\right)\\
        &=  2^{3r+2} \cdot \left(\frac{2}{3} - \frac{1}{2^{r+1}} + \frac{1}{3 \cdot 2^{2r+2}}\right) = \frac{8\cdot2^{3r} - 6\cdot2^{2r} + 2^r}{3}.
    \end{split}
    \end{equation*}
\end{proof}

\begin{proposition}
    The lower bound $A_{\dmin}^{\mathrm{LB}}(r)$ given in Theorem~\ref{thm:rmbound} is tight for a pre-transformed polar code $\mathcal P(\mathcal I,\,\boldsymbol T)$ with $\RM(r,n)$ rate-profile $\mathcal{I}$ of order $r = n-2$.
\end{proposition}
\begin{proof}
    For order $r = n-2$ holds $\w(i)=n-r=2$ with $i \in \mathcal{I}_{\dmin}$. In this case, Lemma~\ref{lem:rm non pretransformable} gives that a coset with index support $\mathcal{S}_i = \{x,y\}$, where $0 \leq x < y < n$, is non pre-transformable. Since this describes all coset indices ${i \in \mathcal{I}_{\dmin}}$, we find with $\Admin(\Ci)=2^{\left|\mathcal{I} \cap \mathcal{K}_i\right|}$ that $\Admin(\C) = A_{\dmin}^{\mathrm{LB}}(r)$.
\end{proof}

\ifPagebreaks
\clearpage
\fi

\section{Results}
\label{sec:results}
\begin{table*}[htp]
    \centering
    \resizebox{\linewidth}{!}{
        \footnotesize
        \begin{tabular}{*{5}{crr}}
    \toprule
    & \multicolumn{2}{c}{$r = 2$} && \multicolumn{2}{c}{$r = 3$} && \multicolumn{2}{c}{$r = 4$} && \multicolumn{2}{c}{$r = 5$} && \multicolumn{2}{c}{$r = 6$}\\
    \cmidrule{2-3} \cmidrule{5-6} \cmidrule{8-9}  \cmidrule{11-12} \cmidrule{14-15}
    $n$ & $p(x)$ & $A_{d_\mathrm{min}}$ && $p(x)$ & $A_{d_\mathrm{min}}$ && $p(x)$ & $A_{d_\mathrm{min}}$ && $p(x)$ & $A_{d_\mathrm{min}}$ && $p(x)$ & $A_{d_\mathrm{min}}$\\ 
    \midrule
    5 & $1027_8$ & 236\\ 
    6 & $400115_8$ & 252 && $1027_8$ & 2136\\
    7 & $410073_8$ & 260 && $400115_8$ & 2136 && $2724313_8$ & 13920\\
    8 & $410073_8$ & 292 && $410073_8$ & 2152 && $2724313_8$ & 13920 && $4347071_8$ & 98200\\
    9 & $410073_8$ & 424 && $410073_8$ & 2300 && $2724313_8$ & 13968 && $4347071_8$ & 98200 && $5767471_8$ & 737496\\
    10 & $410073_8$ & 952 && $410073_8$ & 3584 && $7021047_8$ & 14604 && $4347071_8$ & 98264 && $5767471_8$ & 737496\\ 
    11 & $\phantom{0}410073_8$ & \phantom{00}3048 && $\phantom{0}410073_8$ & \phantom{0}14208 && $7021047_8$ & \phantom{0}25936 && $4347071_8$ & 100900 && $5767471_8$ & 737624\\
    \bottomrule
    \label{tab:PACpolynomials}
\end{tabular}
    }
    \caption{\footnotesize Polynomials $p(x)$  that achieve the minimum $A_{d_\mathrm{min}}$ of all $\operatorname{deg}(p(x)) \leq 20$ polynomials for \ac{PAC} codes with different $\RM\left(r,n\right)$ rate-profiles.}
    \label{tab:pac_poly}
    \vspace{-.5cm}
\end{table*}

\subsection{Algorithmic Complexity}
\begin{figure}[htp]
    \centering
    \resizebox{\columnwidth}{!}{\begin{tikzpicture}
\begin{axis}[
	width=\linewidth,
	height=0.67\linewidth,
	scale=1,
	grid style={dotted,gray},
	xmajorgrids,
	yminorticks=true,
	ymajorgrids,
	yminorticks=true,
	tick align=outside,
	tick pos=left,
	legend columns=1,
	legend pos=north west,   
	legend cell align={left},
	legend style={at={(axis description cs:0.01,0.6)},fill,fill opacity=0.75, text opacity=1, draw=none},
    grid style={on layer=axis background},
	xtick={5,7,...,13},
	tick style={color=black},
	xlabel={$n = \log_2{N}$},
	legend image post style={mark indices={}},
	ymode=log,
	mark size=1.5pt,
	xmin=5,
	xmax=13,
    ymin=236, %
	ymax=1e18, %
]

\addplot [color=apfelgruen,densely dotted,line width=\lineWidth,mark=x,mark size=\markSize,mark options={solid}] %
table[col sep=comma]{
5, 4806
7, 2039786
9, 1992493331
11, 7921316325039
13, 128488246466532412
};
\label{plot:Pre-transform checks RS}

\addplot [color=apfelgruen,densely dashed,line width=\lineWidth,mark=x,mark size=\markSize,mark options={solid}] %
table[col sep=comma]{
5, 2115
7, 874518
9, 940297767
11, 3262281615573
13, 38368590563625357
};
\label{plot:Tree updates RS}

\addplot [color=apfelgruen,solid,line width=\lineWidth,mark=x,mark size=\markSize,mark options={solid}] %
table[col sep=comma]{
5, 284
7, 87294
9, 52823191
11, 113560498085
13, 952165734799676
};
\label{plot:Visited cosets RS}

\addplot [color=mittelblau,densely dashed,line width=\lineWidth,mark=o,mark size=\markSize,mark options={solid}] %
table[col sep=comma]{
5, 717
7, 27140
9, 327039
11, 7277071
13, 3294314916
};
\label{plot:Tree updates EL}

\addplot [color=mittelblau,solid,line width=\lineWidth,mark=o,mark size=\markSize,mark options={solid}] %
table[col sep=comma]{
5, 295
7, 4976
9, 32493
11, 327168
13, 30544839
};
\label{plot:Visited cosets EL}

\addplot [color=pink,densely dotted,line width=\lineWidth,mark=+,mark size=\markSize,mark options={solid}] %
table[col sep=comma]{
5, 398
7, 12912
9, 101416
11, 2188763
13, 1288386923
};
\label{plot:Pre-transform checks proposed}

\addplot [color=pink,densely dashed,line width=\lineWidth,mark=+,mark size=\markSize,mark options={solid}] %
table[col sep=comma]{
5, 345
7, 6642
9, 49383
11, 508586
13, 41072372
};
\label{plot:Tree updates proposed}

\addplot [color=pink,solid,line width=\lineWidth,mark=+,mark size=\markSize,mark options={solid}] %
table[col sep=comma]{
5, 284
7, 4836
9, 31239
11, 316734
13, 30459464
};
\label{plot:Visited cosets proposed}

\addplot [color=black,dotted,line width=\lineWidth,mark=Mercedes star,mark size=\markSize,mark options={solid}]
table[col sep=comma]{
5,620
7,94488
9,52955952
11,113562778208
13,952165772592320
15,31566670174891755904
17,4161765486531358105244416
19,2188358249970312675214379656704
};
\label{plot:Admin RM plain}

\addplot [color=black,dotted,line width=\lineWidth,mark=Mercedes star flipped,mark size=\markSize,mark options={solid}]
table[col sep=comma]{
5, 236
7, 2136
9, 15216
11, 103148
13, 1528328
15, 872804488
};
\label{plot:Admin PAC}

\newcommand\testTest{\widthof{content}}

\coordinate (legend) at (axis description cs:0.01,1);
\coordinate (legend Awmin) at (axis description cs:0.01,0.59);
\end{axis}

\matrix [
        draw=none,
        row sep=-\pgflinewidth,
        matrix of nodes,
        nodes={anchor=center},
        anchor=north west,
        font=\footnotesize,
        inner sep=2pt,
        column 4/.style={nodes={anchor=west}},
    ] (mat) at (legend) {
        \hspace{1sp}\cite{Rowshan2023Impact}\vphantom{p} & \hspace{1sp}\cite{PartialEnumPAC}\vphantom{p} & Prop. & Metric\vphantom{p} \\
        \hline
        \ref{plot:Pre-transform checks RS}\vphantom{f} & \ref{plot:Tree updates EL}\vphantom{f} & \ref{plot:Pre-transform checks proposed}\vphantom{f} & Pre-transform checks\vphantom{f}\\
        \ref{plot:Tree updates RS}\vphantom{Mp} & \ref{plot:Tree updates EL}\vphantom{Mp} & \ref{plot:Tree updates proposed}\vphantom{Mp} & Message updates\\
        \ref{plot:Visited cosets RS}\vphantom{V} & \ref{plot:Visited cosets EL}\vphantom{V} & \ref{plot:Visited cosets proposed}\vphantom{V} & Visited sub-trees\vphantom{V}\\
    };

\matrix [
        draw=none,
        row sep=-\pgflinewidth,
        matrix of nodes,
        nodes={anchor=west},
        anchor=north west,
        font=\footnotesize,
        inner sep=2pt,
    ] (mat2) at (legend Awmin) {
        $\Admin$ \\
        \hline
        \ref{plot:Admin RM plain} Plain\\
        \ref{plot:Admin PAC} PAC\\
    };
\end{tikzpicture}}
    \caption{\footnotesize Complexity comparison of the $\dmin$-weight codeword enumeration for $\RM(\nicefrac{(n-1)}{2},n)$ rate-profile \ac{PAC} codes with polynomial $p(x)\mathrel{\widehat{=}}5767471_8$. Note that the pre-transform checks coincide with the message updates in \cite{PartialEnumPAC}.}
    \label{fig:PAC_complexity_comparison}
\end{figure}
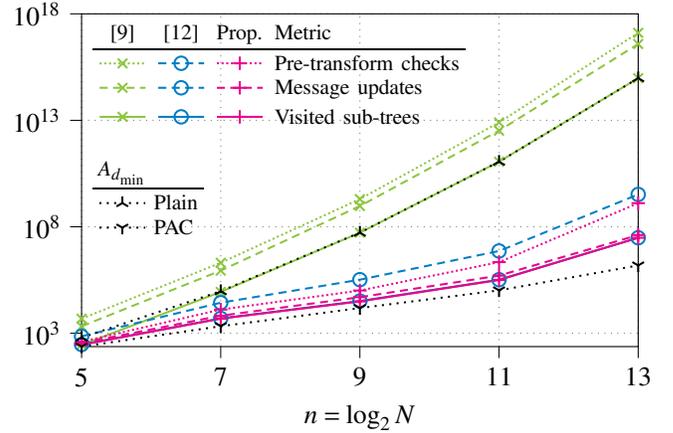
First, we evaluate the computational complexity of the proposed minimum weight codeword enumeration algorithm (Algorithm~\ref{alg:error_coeff}) and compare it to the algorithms proposed in \cite{Rowshan2023Impact} and \cite{PartialEnumPAC}. To this end, we count the number of visited sub-trees in the search tree, the number of checks, whether a message fulfills the precoding constraints, and the number of message updates. Fig.~\ref{fig:PAC_complexity_comparison} shows these three metrics for rate-$\nicefrac{1}{2}$ \ac{PAC} codes with \ac{RM} rate-profile and precoding polynomial $p(x)\mathrel{\widehat{=}}5767471_8$ in terms of the logarithmic block length. For reference, we also plot the number of minimum weight codewords  $A_{d_\mathrm{min}}$ for both the plain \ac{RM} code and the resulting \ac{PAC} code.
As we can see, in all metrics, our proposed algorithm is several orders of magnitude less complex than the one from \cite{Rowshan2023Impact}. The difference is more significant for long block lengths. For example, for the $(8192,4096)$ code, the number of pre-transform checks is reduced by a factor of $10^8$. Compared to \cite{PartialEnumPAC}, the proposed method mainly reduces the number of message updates.
The complexity of the proposed algorithm scales roughly with $A_{d_\mathrm{min}}$ of the pre-transformed code, rather than the plain code.

Computing $\Admin=3120$ of the \ac{PAC} code proposed in \cite{arikan2019pac} with $\RM(3,7)$ rate-profile and  $p(x)\mathrel{\widehat{=}}155_8$ takes about \qty{1}{ms} with our nonparallel C implementation of Algorithm~\ref{alg:error_coeff}.%

\subsection{Number of Minimum Weight Codewords of \acp{PTPC}}
Using  Algorithm~\ref{alg:error_coeff}, we find \ac{PAC} codes with \ac{RM} rate-profiles that are optimal with respect to the number of minimum weight codewords. 
For $2\le r < n-2$, $5 \le n \le 11$, optimal convolutional polynomials $p(x)$ of maximum degree 20 were found using exhaustive search and listed in octal notation in Table~\ref{tab:pac_poly}. 
If multiple polynomials achieved the minimum $\Admin$, we list the one with the lowest degree and fewest non-zero coefficients.
Following Corollary~\ref{cor:dmin rm}, the codes have the same minimum distance as the plain \ac{RM} codes, i.e., $d_\mathrm{min}=2^{n-r}$.

\begin{figure}[htp]
	\centering
    \resizebox{\columnwidth}{!}{\begin{tikzpicture}
\begin{axis}[
    width=\linewidth,
    height=0.67\linewidth,
    scale=1,
    grid style={dotted,gray},
    xmajorgrids,
    yminorticks=true,
    ymajorgrids,
    yminorticks=true,
    tick align=outside,
    tick pos=left,
    legend columns=1,
    legend pos=north west,   
    legend cell align={left},
    legend style={at={(axis description cs:0.01,0.99)},fill opacity=0.75, text opacity=1, draw=none, row sep=0pt, inner sep=2pt},
    xtick={5,7,...,21},
    tick style={color=black},
    xlabel={$n = \log_2{N}$},
    ylabel={$\Admin$},
    legend image post style={mark indices={}},
    ymode=log,
    mark size=1.5pt,
    xmin=5,
    xmax=21,
    ymin=64, %
    ymax=1e10, %
]

\path [fill=violettblau!35]
(axis cs:5,620)
--(axis cs:7,3610)
--(axis cs:9,17588)
--(axis cs:11,103038)
--(axis cs:13,743627)
--(axis cs:15,5749458)
--(axis cs:17,45331203)
--(axis cs:19,360230562)
--(axis cs:21,2872468478)
--(axis cs:21,2872465657)
--(axis cs:19,360223517)
--(axis cs:17,45320317)
--(axis cs:15,5741105)
--(axis cs:13,739555)
--(axis cs:11,99770)
--(axis cs:9,14592)
--(axis cs:7,2043)
--(axis cs:5,236)
--cycle;

\addplot [color=violettblau,line width = 0.75pt,mark=+,mark size=\markSize,mark options={solid}, forget plot] %
table[col sep=comma]{
5,378.00
7,2766.74
9,15936.22
11,101272.50
13,741520.24
15,5745116.69
17,45326196.69
19,360227039.5
21,2872467067.5
};
\label{plot:simulation_rm}

\addplot [color=apfelgruen,line width = 0.75pt,mark=square,mark size=\markSize,mark options={solid}, densely dashed, forget plot] %
table[col sep=comma]{
5,380.0
7,2766.90625
9,15936.337757110596
11,101280.007599038
13,741520.0000943244
15,5745060.000000669
17,45326388.5
19,360222723.28125
21,2872471680.102539
};
\label{plot:probabilistic_rm}

\addplot [color=pink,line width = 0.75pt,mark=x,mark size=\markSize,mark options={solid}, dashdotted, forget plot] %
table[col sep=comma]{
5,140
7,1240
9,10416
11,85344
13,690880
15,5559680
17,44608256
19,357389824
21,2861214720
};
\label{plot:lowerbound_rm}

\addplot [color=anthrazit!67,line width = 0.75pt,mark=*,mark size=1pt,mark options={solid}, loosely dotted, forget plot] %
table[col sep=comma]{
5,64
7,512
9,4096
11,32768
13,262144
15,2097152
17,16777216
19,134217728
21,1073741824
};
\label{plot:OrderAnalysis}

\addplot [color=black,line width=\lineWidth,mark=Mercedes star,mark size=\markSize,mark options={solid},  dotted]
table[col sep=comma]{
5,620
7,94488
9,52955952
11,113562778208
13,952165772592320
15,31566670174891755904
17,4161765486531358105244416
19,2188358249970312675214379656704
};
\label{plot:RMplain}
\addlegendentry{\footnotesize Plain};

\addplot [color=black,line width = 0.75pt, dotted, mark=Mercedes star flipped,mark size=\markSize,mark options={solid}] %
table[col sep=comma]{
5,236
7,2136
9,13968
11,100900
};
\label{plot:PACoptimized}
\addlegendentry{\footnotesize Optimized PAC};

\coordinate (legend random precoded) at (axis description cs:0.99,0.01);
\end{axis}

\matrix [
        draw=none,
        text opacity=1,
        row sep=-\pgflinewidth,
        matrix of nodes,
        nodes={anchor=west},
        anchor=south east,
        font=\footnotesize,
        inner sep=2pt,
    ] (mat) at (legend random precoded) {
        Randomly pre-transformed \\
        \hline
        \ref{plot:simulation_rm} Computation average\\
        \ref{plot:probabilistic_rm} Probabilistic \cite{Li2021Weightspectrum} \\
        \ref{plot:lowerbound_rm} Lower bound (\ref{eq:lowerbound_rm})\\
        \ref{plot:OrderAnalysis} Order analysis \cite{li2023weightspecturmimprovement}\\
    };

\end{tikzpicture}}
	\caption{\footnotesize Number of $d_\mathrm{min}$-weight codewords of convolutional (\ac{PAC}) and randomly pre-transformed $\RM(\nicefrac{(n-1)}{2},n)$ codes with $32 \le N \le 2097152$. 
    }
	\label{fig:RandomPrecoding}
\end{figure}
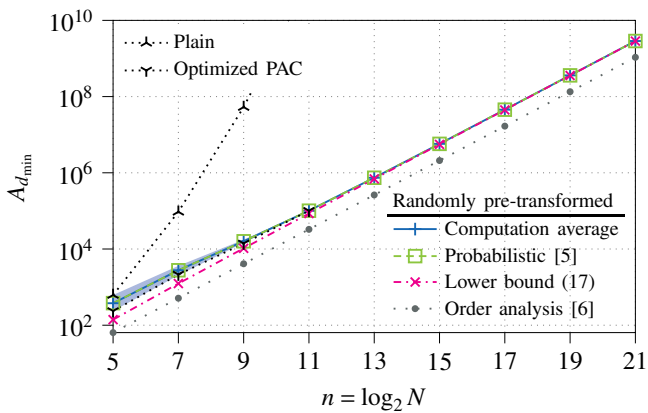
Next, we compare various precoder designs for \acp{PTPC} with \ac{RM} rate-profiles. Fig.~\ref{fig:RandomPrecoding} shows the error coefficient $\Admin$ for convolutional and random precoders, as well as the related bounds for rate-$\nicefrac{1}{2}$ \ac{RM} codes with $5\le n \le 21$.
We randomly generated at least 100 pre-transform matrices for $n<19$ and two for $n\geq19$ and computed $\Admin$ using Algorithm~\ref{alg:error_coeff}. We plot the range of the resulting values and their average. We see that the results match precisely the predictions of the probabilistic method from \cite{Li2021Weightspectrum} and thus, verify their accuracy. Moreover, the lower bound given in (\ref{eq:lowerbound_rm}) closely approaches these two curves for large $N$. The order analysis given in \cite{li2023weightspecturmimprovement} predicts the correct scaling, however consistently underestimates $\Admin$.
We want to emphasize that it is exactly the complexity reduction of our proposed algorithm that allows the explicit computation of the number of minimum weight codewords for specific pre-transformations for codes as long as $N=2097152$, which was not computationally feasible before.
The optimized \ac{PAC} codes (listed in Table~\ref{tab:pac_poly}) have a comparably low number of $\dmin$-weight codewords and the curve lies on top of the minimum achieved using random pre-transformations.

\ifPagebreaks
\newpage
\fi
\ifPagebreaks
\clearpage
\fi
\section{Conclusion}\label{sec:conc}
In this paper, we proposed a novel, low-complexity algorithm for enumerating minimum-weight codewords of \acp{PTPC} based on tree intersection.
The algorithm shows significantly lower complexity than the state-of-the-art methods described in \cite{Rowshan2023Impact} and \cite{PartialEnumPAC}.
Therefore, it enables for the first time to directly take the actual number of minimum-weight codewords into account in the code design. As an application, optimal convolutional polynomials for \ac{PAC} codes have been determined through a now possible exhaustive search.

\bibliographystyle{IEEEtran}
\bibliography{main}

\end{NoHyper}
\end{document}